%% file: main.tex
\PassOptionsToPackage{dvipsnames,svgnames}{xcolor}

\documentclass[11pt]{article}

\usepackage{amsmath,amsthm}

  \newtheorem{theorem}{Theorem}[section]

  \newtheorem{lemma}[theorem]{Lemma}
  \newtheorem{corollary}[theorem]{Corollary}
  
  \newtheorem{definition}[theorem]{Definition}
  
\usepackage{natbib}

\usepackage{fullpage}
\usepackage{etoolbox}
\usepackage{hyperref}

    \newcommand{\revise}[1]{#1}
\input{conf/header}

    \newcommand{\lukasnote}[1]{}
    
    \newcommand{\jacobnote}[1]{{}}
    
    \newcommand{\rasmusnote}[1]{{}}

    \newcommand{\remove}[1]{{}}

\input{conf/abbreviations}

\begin{document}

\title{Lower Bounds for Private Graph Optimization Problems using Reconstruction Attacks}

\author{
  Jacob Imola\\
  University of Waterloo\\
  \texttt{jimola@uwaterloo.ca}
  \and
  Rasmus Pagh\\
  BARC, University of Copenhagen\\
  \texttt{pagh@di.ku.dk}
  \and
  Lukas Retschmeier\\
  BARC, University of Copenhagen\\
  \texttt{retschmeier.lukas@gmail.com}
}

\maketitle
\input{content/abstract}
\input{content/introduction}
\input{content/preliminaries}
\input{content/mst}\label{sec:results}
\input{content/worst-case-bounds/minimum-weight-perfect-matching}
\input{content/hc}
\input{content/algorithms}

\input{content/conclusion}
\section{Acknowledgments}
Pagh and Retschmeier were supported by a Data Science Distinguished Investigator grant from Novo Nordisk Fonden, and are part of BARC, supported by the VILLUM Foundation grant 54451. Imola partially completed this work while working at BARC, University of Copenhagen, and was supported by the same grants.

%\input{content/merged/content/results/linfty} % As discussed with R. Move into seperate note 
%\input{content/clustering-lb/introduction}
%\input{content/clustering-lb/lb}
%\input{content/clustering-lb/fw}


\newpage
\bibliographystyle{plainnat}
\bibliography{main}

\newpage
\appendix

\input{content/appendix-full}
\end{document}

%% file: conf/header.tex
\usepackage{algorithm}
\usepackage{algorithmic}
\usepackage{url}
\usepackage{thmtools,thm-restate}
\usepackage[noabbrev,capitalise]{cleveref}
\usepackage{mathtools}
\usepackage{xspace}
\usepackage{verbatim}
\usepackage{mathrsfs}
\usepackage{tabularx}
\usepackage{enumitem}
\usepackage{derivative}
\usepackage{bm}
\usepackage{multirow}
\usepackage{diagbox}
\usepackage{nicematrix}
\usepackage{parskip}
\usepackage{adjustbox}
\usepackage{transparent}
\usepackage[all]{nowidow}

\usepackage{marginnote}
\reversemarginpar    % opposite side

\usepackage{subcaption}

\usepackage{tikz}
\usepackage{tikz-3dplot}
\usepackage{amsfonts}

\usepackage{tikz-3dplot}
\usepackage{conf/tikzit}
\usetikzlibrary{arrows}
\usepackage{pgf}
 \input{conf/tikz.tikzstyles}
\usetikzlibrary{positioning,chains,fit,shapes,calc,er,automata,graphs,graphs.standard,shapes.symbols}
\usepackage[format=plain,
            labelfont={bf,it},
            textfont=it]{caption}
            
\usepackage{graphicx} % Required for inserting images
\usepackage{appendix}

\usepackage{dsfont}

\newcommand*\ie{i.\kern.1em e., }
\newcommand*\eg{e.\kern.1em g., }
\newcommand*\cf{c.\kern.1em f.\ }
\newcommand*\almev{a.\kern.1em e.\ }

\theoremstyle{plain}
\crefname{claim}{Claim}{Claims}
\crefname{fact}{Fact}{Facts}

\theoremstyle{definition}

\newtheorem{observation}{Observation}

\theoremstyle{plain}
  \newtheorem*{theorem*}{Theorem}

\newcommand{\ignore}[1]{}

\newcommand{\Ex}[1]{\bE \left[ #1 \right]}

\renewcommand{\Pr}[1]{\bP \left[ #1 \right]} % Probability
\newcommand{\Pru}[2]{\underset{ #1 }\bP \left[ #2 \right]}

\DeclarePairedDelimiter\abs{\lvert}{\rvert}%
\DeclarePairedDelimiter\norm{\lVert}{\rVert}%

\makeatletter
\let\oldabs\abs
\def\abs{\@ifstar{\oldabs}{\oldabs*}}
\let\oldnorm\norm
\def\norm{\@ifstar{\oldnorm}{\oldnorm*}}
\makeatother

\newcommand{\x}{\ensuremath{\vec{x}}}
\newcommand{\X}{\ensuremath{\vec{x}}}
\newcommand{\Xleqi}{\ensuremath{\vec{x}_{<i}}}
\newcommand{\Xgeqi}{\ensuremath{\vec{x}_{>i}}}

\newcommand{\W}{\ensuremath{\vec{w}}}

\newcommand{\R}{\ensuremath{\mathbb{R}}}

\newcommand{\cA}{\ensuremath{\mathcal{A}}}
\newcommand{\cB}{\ensuremath{\mathcal{B}}}

\newcommand{\cD}{\ensuremath{\mathcal{D}}}
\newcommand{\cE}{\ensuremath{\mathcal{E}}}
\newcommand{\cF}{\ensuremath{\mathcal{F}}}
\newcommand{\cG}{\ensuremath{\mathcal{G}}}

\newcommand{\cO}{\ensuremath{\mathcal{O}}}
\newcommand{\cP}{\ensuremath{\mathcal{P}}}

\newcommand{\cS}{\ensuremath{\mathcal{S}}}
\newcommand{\cT}{\ensuremath{\mathcal{T}}}

\newcommand{\cY}{\ensuremath{\mathcal{Y}}}
\newcommand{\cX}{\ensuremath{\mathcal{X}}}

\newcommand{\bE}{\ensuremath{\mathbb{E}}}

\newcommand{\bP}{\ensuremath{\mathbb{P}}}

\numberwithin{equation}{section}
            
\renewcommand{\vec}[1]{\mathbf{#1}}

\crefname{theorem}{Theorem}{Theorems}
\Crefname{theorem}{Theorem}{Theorems}

\crefname{lemma}{Lemma}{Lemmas}
\Crefname{lemma}{Lemma}{Lemmas}

\crefname{definition}{Definition}{Definitions}
\Crefname{definition}{Defintion}{Definitions}

\colorlet{linkcolour}{DarkBlue}

\hypersetup{colorlinks=true, linkcolor=linkcolour, citecolor=linkcolour, urlcolor=linkcolour,}

\renewcommand{\epsilon}{\varepsilon}

\DeclareMathOperator{\EX}{\mathbb{E}}% expected value
\renewcommand\vec{\mathbf}

\usepackage{bbold}

\newcommand{\PAREN}[1]{{\left( {#1} \right)}}

\newcommand{\mnom}{\ensuremath{\operatorname{MultiNom}}}
\newcommand{\lap}{\ensuremath{\operatorname{Lap}}}

\newcommand{\Bin}{\ensuremath{\operatorname{Bin}}}

\newcommand{\hccost}[1]{\textnormal{cost}_{#1}}

\def\ra{\mathsf{Rec}} % Synthesized with first part.
\def\hc{\mathsf{HC}}
\newcommand{\suc}{\ensuremath{\operatorname{Succ}}}
\def\mstalg{\ensuremath{\cA_{\operatorname{MST}}}}

\newcommand{\mwpmLong}{\textnormal{\textsc{Minimum-Weight Perfect Matching}}\xspace}
\newcommand{\mstLong}{\textnormal{\textsc{Minimum Spanning Tree}}\xspace}
\newcommand{\hcLong}{\textnormal{\textsc{Hierarchical Clustering}}\xspace}
\newcommand{\ssptLong}{\textnormal{\textsc{Single-Source Shortest-Path Trees}}\xspace}

\newcommand{\enc}{\ensuremath{\operatorname{Enc}}}

\newcommand{\Rec}{\ensuremath{\operatorname{Rec}}}
\newcommand{\Aenc}{\ensuremath{\operatorname{Enc}}}
\newcommand{\Adec}{\ensuremath{\operatorname{Dec}}}

    \definecolor{nicepurple}{HTML}{802c7f}
    \definecolor{orangeish}{HTML}{f1a340} %orange
    \definecolor{grayish}{HTML}{f7f7f7} %white-gray
    \definecolor{purpleish}{HTML}{998ec3} %purple
    \definecolor{blueish}{HTML}{004488}
    \definecolor{darkgreen}{RGB}{2,100,64} % Even darker
    \definecolor{lightgreen}{HTML}{b2f2bb}
    \definecolor{lightblueish}{RGB}{230,244,255}

\let\oldvec\vec
\renewcommand{\vec}[1]{\oldvec{\MakeLowercase{#1}}}

%% file: conf/tikz.tikzstyles
\tikzstyle{BLACKLARGE}=[draw=black, shape=circle, fill=black, inner sep=3pt]
\tikzstyle{BLACK}=[draw=black, shape=circle, fill=black, inner sep=1.5pt]
\tikzstyle{DOM}=[fill={rgb,255: red,122; green,0; blue,42}, draw=black, shape=circle, inner sep=3pt, line width=0.5px]
\tikzstyle{NONE}=[fill={rgb,255: red,239; green,66; blue,109}, draw=black, shape=circle, inner sep=3pt, line width=0.5px]
\tikzstyle{NTWO}=[fill={rgb,255: red,111; green,176; blue,118}, draw=black, shape=circle, inner sep=3pt, line width=0.5px]
\tikzstyle{NTHR}=[fill={rgb,255: red,252; green,199; blue,18}, draw=black, shape=circle, inner sep=3pt, line width=1pt]
\tikzstyle{GRAYN}=[fill={rgb,255: red,213; green,213; blue,213}, draw=black, shape=circle, inner sep=1.5pt, line width=1pt]
\tikzstyle{BRGRAY}=[fill={rgb,255: red,236; green,236; blue,236}, draw=none, shape=circle, inner sep=3pt, line width=1pt]
\tikzstyle{LIGHTRED}=[fill={rgb,255: red,199; green,83; blue,48}, draw=black, shape=circle, inner sep=3pt, line width=1pt]

\tikzstyle{TEXTSTD}=[fill=white, draw=black, shape=rectangle, tikzit shape=rectangle, text width=3cm, rounded corners]
\tikzstyle{TEXTHIGH}=[fill={rgb,255: red,231; green,245; blue,255}, draw=black, shape=rectangle, tikzit shape=rectangle, text width=3cm, rounded corners]
\tikzstyle{TEXTHIGHGRAY}=[fill={rgb,255: red,220; green,220; blue,220}, draw=black, shape=rectangle, tikzit shape=rectangle, text width=3cm, rounded corners]
\tikzstyle{TEXTHIGHGREEN}=[fill={rgb,255: red,149; green,242; blue,145}, draw=black, shape=rectangle, tikzit shape=rectangle, text width=3.4cm, rounded corners]
\tikzstyle{TEXTHIGHYELLOW}=[fill={rgb,255: red,250; green,250; blue,160}, draw=black, shape=rectangle, tikzit shape=rectangle, text width=3cm, rounded corners]

\tikzstyle{EDGE}=[-, fill=none, line width=1.5pt]
\tikzstyle{BLUE}=[-, draw={rgb,255: red,0; green,101; blue,189}, line width=1.5px]
\tikzstyle{GRAY}=[-, fill={rgb,255: red,128; green,128; blue,128}]
\tikzstyle{DARKGREEN}=[-, fill=none, draw={rgb,255: red,111; green,176; blue,118}, line width=1.25pt]
\tikzstyle{big dash}=[-, dashed, dash pattern=on 4mm off 2mm, fill={rgb,255: red,178; green,255; blue,253}]
\tikzstyle{big dash thick}=[-, thick, dashed, dash pattern=on 4mm off 2mm, fill={rgb,255: red,178; green,255; blue,253}]
\tikzstyle{new edge style 0}=[-, fill={rgb,255: red,230; green,255; blue,252}]
\tikzstyle{FILLGREEN}=[-, fill={rgb,255: red,111; green,176; blue,118}]
\tikzstyle{FILLBLUE}=[-, fill={rgb,255: red,247; green,249; blue,255}, line width=0.9px]
\tikzstyle{FILLPURPLE}=[-, fill={rgb,255: red,255; green,233; blue,239}]
\tikzstyle{FILLDARKBLUE}=[-, fill={rgb,255: red,238; green,230; blue,255}]
\tikzstyle{FILLDARKBLUEINVISIBLE}=[-, draw=none, fill={rgb,255: red,238; green,230; blue,255}]
\tikzstyle{FILLGREENINVISBLE}=[-, draw=none, fill={rgb,255: red,232; green,255; blue,207}]
\tikzstyle{ULTRAGRAY}=[-, draw={rgb,255: red,213; green,213; blue,213}, line width=0.9px]

\tikzstyle{EDGEDASHED}=[-, dashed, fill=none, line width=1.3pt]
\tikzstyle{simple directed edge}=[->, draw=black, thick, line width=1.5pt]
\tikzstyle{simple dashed directed edge}=[->, draw=black, dashed]
\tikzstyle{SLIGHTLYDASHED}=[-, dotted, fill=none, draw={rgb,255: red,128; green,128; blue,128}]
\tikzstyle{simple dashed directed edge not thick}=[->, draw=black, line width=0.45mm, dashed]
\tikzstyle{simple}=[->, draw=black, thick]
\tikzstyle{simple}=[draw=black, fill=none, tikzit draw=black, ->]
\tikzstyle{invis}=[-, draw=none]
\tikzstyle{g5edge}=[-, line width=1pt]
\tikzstyle{DashedBlue}=[-]
\tikzstyle{new edge style 1}=[-, fill=none, draw={rgb,255: red,0; green,99; blue,185}, line width=1pt, dashed]
\tikzstyle{new edge style 2}=[-, draw={rgb,255: red,252; green,199; blue,18}, line width=1.5pt]
\tikzstyle{NiceRed}=[-, fill=none, draw={rgb,255: red,199; green,83; blue,48}, line width=0.75pt]
\tikzstyle{NiceRedThick}=[-, fill=none, draw={rgb,255: red,199; green,83; blue,48}, line width=1.2pt]

\tikzstyle{thick}=[-, fill=none, line width=1.25pt]
\tikzstyle{noise arrow stopper}=[-, line width=0.35mm, draw={rgb,255: red,116; green,66; blue,128}]
\tikzstyle{black dotted}=[-, dotted, line width=1pt]
\tikzstyle{new edge style 1}=[line width=0.5mm, draw=black, ->, fill=none]
\tikzstyle{gray edge}=[-, fill=none, draw={rgb,255: red,191; green,191; blue,191}, line width=0.75pt]
\tikzstyle{striped-ultragray}=[-, draw={rgb,255: red,64; green,64; blue,64}, dotted, fill=white]
\tikzstyle{FILLLIGHTRED}=[-,draw={rgb,255: red,191; green,0; blue,64},line width=0.75pt, fill={rgb,255: red,255; green,239; blue,234}]
\tikzstyle{FILLLIGHTREDB}=[-,draw={rgb,255: red,191; green,0; blue,64},line width=1.5pt, fill={rgb,255: red,255; green,239; blue,234}] % BIGGER VERSION

%% file: conf/abbreviations.tex
\newcommand{\probp}{\ensuremath{ \frac{0.5}{d_{\max}}}}

%% file: content/abstract.tex
\begin{abstract}
    This paper studies fundamental graph optimization problems under differential privacy (DP) and shows new, reconstruction-based lower bounds.
    We consider a graph $G = (V, E, \vec{w})$ where the vertex set $V$ and edges $E$ are public and the weights $\vec{w}:E\rightarrow \R$ must be kept differentially private under an $\ell_1$ neighboring relation.
    
    For the problems of releasing a minimum-weight spanning tree and a minimum-weight perfect matching, we show new, tight error bounds of $\Omega(n\cdot\log(m/n)/\epsilon)$ on worst-case graphs with $n$ vertices and $m>2n$ edges.
    The upper bounds are known pure DP algorithms while the new lower bound holds even under approximate $(\varepsilon,\delta)$-DP as long as $\delta \leq (n/m)^{\Omega(1)}$.
    Our lower bounds improve the $\Omega(n/\epsilon)$ lower bounds of Sealfon (PODS~'16).
    The fact that approximate DP does not reduce error for MST under the $\ell_1$ neighboring relation contrasts with the recent upper bound of Pagh et al. (PODS~'25) which shows that approximate DP allows much better error under the $\ell_\infty$ neighboring relation.

    Going beyond worst-case graphs, we give lower bounds for large families of sparse graphs with expansion properties. We show a lower bound of $\Omega(n / \epsilon)$ for the minimum spanning tree for any graph where the minimum cut is at least $\Omega(\log(n))$. 
    Finally, we consider the problem of private hierarchical clustering under Dasgupta's cost function (STOC~'16) and show the first approximate DP lower bound parameterized by the minimum weight of a balanced cut. This extends lower bounds of Deng et al. (ICLR~'25) to general graphs and to approximate DP.

    %, i.e. two graphs $G = (V, E, \vec{w})$ and $G' = (V, E, \vec{w'})$ are neighboring if $\norm{\vec{w}-\vec{w'}}_1 \leq \Delta_1$ for some sensitivity parameter $\Delta_1$. 
    
%Our lower bound is based on encoding a private binary vector into the edges of a perfect matching.
%Assuming this parameter grows with $\Omega(n)$, we can reconstruct a large fraction of the original data if the given hierarchical clustering tree with Dasgupta cost less then $o(n^2)$, thus contradicting privacy.
\end{abstract}

%% file: content/introduction.tex
\section{Introduction}

Graph-structured data is a fundamental abstraction in modern data management.
Weighted graphs naturally arise, for example, from relational data describing transportation networks, communication infrastructure, recommendation systems, similarity relations, and interactions between entities.
Many important analysis tasks on such data reduce to classical graph optimization problems, including computing sparse connectivity structures such as minimum spanning trees, assignment structures such as minimum-weight perfect matchings, and multiscale summaries such as hierarchical clusterings.
These structures serve as compact summaries that support downstream tasks such as routing, indexing, similarity search, visualization, and unsupervised learning.

In many applications, however, the edge weights encode sensitive aggregate information derived from individuals.
For example, edge weights may represent traffic intensities collected from users, communication frequencies between populations, similarity scores derived from behavioral data, or compatibility measures computed from medical records.
In such settings, a single individual may affect multiple edge weights simultaneously, but only by a small amount in each coordinate.
%Compared to edge-level privacy models that treat edges as discrete records, the $\ell_1$ model naturally captures continuous-valued and aggregate data contributions, making it particularly well suited for weighted graph optimization problems.
This motivates studying differential privacy under an $\ell_1$ neighboring relation on \emph{edge weights}, where neighboring weighted graphs differ by bounded total change in their edge-weight vectors\revise{~\citep{Sealfon_2016}}.
Other researchers have applied this appealing privacy model to other graph problems, including hierarchical clustering~\citep{deng2025pricedifferentialprivacyhierarchical} and all pairs shortest \revise{distances}~\citep{Chen2023,greg_2024}.
%More recently, private graph optimization under alternative neighboring relations has received renewed attention, including nearly optimal algorithms for MST under $\ell_\infty$ sensitivity~\citep{pagh_2025} and algorithms and lower bounds for private hierarchical clustering~\citep{deng2025pricedifferentialprivacyhierarchical}.

{\it Examples.}
In a transportation network, an MST can represent a sparse backbone connecting cities or distribution centers using low-cost routes inferred from aggregate traffic data.
Publishing such a structure enables routing, infrastructure planning, and resilience analysis while dramatically compressing the underlying weighted network.
Similarly, in a recommendation or similarity graph, a hierarchical clustering, also known as a dendrogram, provides a navigable taxonomy of users, products, or documents derived from sensitive behavioral correlations.
These structures often involve sensitive data, such as road traffic data or user browsing activity, and thus are good candidates for private data release.
%The resulting hierarchy is itself a data product: it can support search, exploration, indexing, and coarse-to-fine recommendation without exposing the full matrix of pairwise similarities.
%More generally, graph optimization outputs often serve as compact, interpretable summaries of weighted relational data, making them particularly attractive targets for private data release.
%\subsection{Background}
%\paragraph{Setting} In this work, we consider \emph{edge-weight differential privacy} proposed by \cite{Sealfon_2016}, where the graph topology with $n$ vertices and $m$ edges itself $G = (V, E)$ is assumed to be publicly known, but the weight vector $\vec{W}$ should be kept private.
%We consider both the $\ell_1$ neighboring relationship, where neighboring graphs $G = (V, E, \W)$ and $G' = (V,E, \W')$ can differ by either $\|\W-\W'\|_1 \leq 1$.
%In this work we prove new tight bounds for these problems where it turns out that the bounds can be achieved by creating a private synthetic graph.
%Our work is thus guided by the following question:%, considering both the $\ell_1$ and $\ell_\infty$ neighboring relationships:

\revise{
\paragraph{Problem Setting}
In this paper, we study the \revise{problems of} \mstLong~({\bf MST}), \mwpmLong~({\bf MWPM}), and \hcLong~({\bf HC}) in the edge-weight model of differential privacy (DP), where the graph topology $G = (V, E)$ with $n$ vertices and $m$ edges is publicly known, but the edge weights $\W \in \R^m$ must be kept private.
A randomized algorithm $\cA$ must output a solution (e.g., a spanning tree, perfect matching, or dendrogram) of the known topology $G$ that approximately minimizes a problem-specific cost function with respect to the private weights. 
The privacy constraint requires that the output distribution of  $\cA$ is similar for neighboring weight $\W, \W'$ vectors that are close in $\ell_1$ distance.
We measure the expected error additively as the difference between the cost of a released solution and the optimum.
We defer a formal introduction to \cref{sec:preliminaries}.
}

\subsection{Our Contributions}
We derive \emph{two} types of \revise{reconstruction-based} lower bounds for \revise{approximate DP and the $\ell_1$ neighboring relationship.}
\revise{We first consider lower bounds on explicit \emph{worst-case} topologies.}
For MST and MWPM, we provide lower bounds showing that the expected additive error must grow with a factor of $\Omega\big(n \cdot  \log (\frac{m}{n})\big)$.
\revise{These lower bounds work via a new reconstruction-attack setup using a dataset from a large domain where more information can be reconstructed compared to a binary dataset.}

This lifts the prior lower bound of $\Omega(n)$ by a logarithmic factor and tightly \revise{matches the upper bound obtained by the simple input perturbation algorithm that adds noise to each weight independently~\citep{aamand-personal-communication,Sealfon_2016}.
}
The fact that input perturbation is an optimal approach for these problems is additional evidence that input perturbation is often a good approach to graph optimization problems, as recently seen in \citet{pagh_2025}.
Furthermore, the fact that this upper bound under pure differential privacy ($\delta = 0$) can be matched by lower bounds shown under approximate differential privacy (with polynomially small $\delta$) is somewhat surprising given the polynomial separation \revise{by a factor of $\cO(\sqrt{n})$} between pure and approximate differential privacy for \mstLong under the $\ell_\infty$ neighboring relation~\citep{pagh_2025}.

\revise{The lower bounds in the previous paragraph, and indeed many other lower bounds in edge-weight DP, use worst-case topologies that are not always realistic in a real-world setting. 
For example, the lower bounds for MST (with $\delta > 0$) and HC use a doubly-connected star graph or the complete graph as their worst-case topology, respectively (see \cref{tab:results}), leaving it open whether better performance is possible on sparser or more realistic topologies.}

\revise{Our second type of lower bounds show that it is not possible to achieve much better performance for large classes of sparse topologies.} 
For MST, we show that \revise{\emph{any}} topology with \emph{minimum cut} at least $\Omega\big(\log(n)\big)$ must incur error $\Omega(n)$, differing from the lower bound of the worst-case topology by a $\log(\frac{m}{n})$ factor. \revise{This complements a lower bound of~\citet{hladik_2024} which establishes that the error must grow linearly with the diameter of the topology but holds for pure DP only.} 
For HC under Dasgupta's cost function \citep{dasgupta2016cost}, we obtain a fully-parameterized lower bound  that grows with the expansion (in the sense of expander graphs), and decreases with the maximum degree of the graph. 
For constant-degree regular expander graphs, our lower bound is $\Omega(n^2)$, which matches the state-of-the-art lower bound of~\citet{deng2025pricedifferentialprivacyhierarchical, pmlr-v202-imola23a} but generalizes it to approximate DP and beyond the complete graph.

%We revisit two graph optimization problems considered in~\citet{Sealfon_2016}, 
%\mstLong~({\bf MST}) and \mwpmLong~({\bf MWPM}) and obtain new lower bounds that match Sealfon's upper bound for sufficiently dense, worst-case graphs.
%In fact, the lower bound matches a recent improvement of Sealfon's upper bounds~\citep{aamand-personal-communication} for all graph densities.
%In addition, we consider the \hcLong({\bf HC}) problem introduced in a seminal paper by~\citet{dasgupta2016cost}, and recently studied under privacy constraints by~\citet{deng2025pricedifferentialprivacyhierarchical}, extending their lower bound to hold for approximate differential privacy.
A summary of all our new results together with known bounds appears in \cref{tab:results}.
Note that the error bounds depend only on graph parameters, so the error is independent of the weight (or cost) of a solution.
Thus these problems can be solved privately with good error bounds whenever the weight (or cost) of an optimal solution is sufficiently large compared to the error bound.

\input{tables/results-l1}

%In addition to considering worst-case graphs, we consider lower bounds that hold for any graph in broad classes such as graphs that do not have a small min-cut and graphs with expansion properties.
%In particular, for \hcLong we show that the lower bound of~\citet{deng2025pricedifferentialprivacyhierarchical} can holds even for sparse expander graphs under approximate differential privacy.

\paragraph{Techniques} 

%Similar to the approach of \citet{Sealfon_2016} our lower bounds work by demonstrating that any algorithm with low error has the property that information carefully encoded into the graph weights can be partially reconstructed from its output.

Our lower bounds work via \emph{reconstruction attacks}, where a private dataset is encoded into the weights of the graph topology, and based on the solution to the graph problem, a decoding process attempts to reconstruct the dataset. If the error of the graph algorithm is low enough, we show that the reconstruction will be successful, contradicting DP. \revise{We introduce several new techniques in designing these reconstruction attacks.}

\revise{For our lower bounds on worst-case topologies, our sharper bounds come from encoding a dataset in $[m/n]^n$ rather than a binary dataset, which allows for more information to be reconstructed. To do this, we need to use a new encoding scheme on a dense topology.
For \mstLong, we use a complete bipartite graph whose left vertices represent coordinates and whose right vertices represent possible values, assigning low weight to the edge connecting each coordinate $i$ to its value $x_i$. 
Thus, a sufficiently accurate spanning tree reveals the value of many coordinates, contradicting the reconstruction guarantees imposed by differential privacy.
Using the same construction again for \mwpmLong, we note that a \emph{perfect} matching can not encode the same value multiple times.
Therefore, we draw a connection to the classic balls-into-bins problem showing that a large fraction of the dataset is collision-free with high probability.
It turns out to be sufficient to run the reconstruction attack on this non-colliding subset of the random dataset.
}

As recently observed by~\citet{aamand-personal-communication} the \revise{known} lower bounds for MST and MWPM are matched by Sealfon's input perturbation algorithms that simply add Laplace noise to each edge weight and run the non-private optimization algorithm as post-processing.
\revise{ For completeness, we include a simple proof of the improved bound of~\citet{aamand-personal-communication}.
\revise{Interestingly, our worst-case lower bound asymptotically matches this upper bound, even when parameterized by the number of vertices \emph{and} edges.}
}

For our second type of lower bound, we focus on reconstructing binary datasets and encode them \emph{randomly} into the edge weights \revise{of a given topology}. 
This requires analyzing the behavior of the \revise{random weights} using ideas from random graph theory. 
\revise{At a high level, we need to establish that a low-cost optimal solution exists in the random weights, and that the decoding procedure will reveal many more $0$-edges than $1$-edges.}
For \mstLong, \revise{it is enough to assume that the topology has min-cut at least $\log(n)$; a union bound over all cuts in the graph will then ensure that a MST of weight $0$ exists. 
Then, a simple decoding procedure where the algorithm guesses all data points in the tree have weight $0$ is an effective reconstruction attack.} 

\revise{For \hcLong, we have to further subsample the edges at a rate of $\frac{1}{d}$ from the topology (assuming the topology is $d$-regular) before encoding the dataset; this ensures that a low-cost clustering exists. 
Given a returned HC, the decoder guesses that all data points in edges crossing the top-most balanced cut (\cref{lem:hc-to-bc}) of the HC are $0$. 
This opens up an interesting challenge where it is possible that returned low-cost HC adversarially has few sampled $0$-weight edges crossing its balanced cut, which would reveal too little information in the dataset. 
To circumvent this, we leverage the fact that the HC algorithm cannot distinguish between sampled $0$-weight edges and non-sampled edges, meaning that any tree it returns cannot adversarially avoid sampled edges, and it will reveal sufficiently many $0$s in the dataset.}

%\lukas{((LR: TODO: Highlight novel insights in HC: (e.g. using a random sparse encoding to enable cost analysis with random graph theory and leveraging conditional independence between zero-weight edges and the returned tree) which to our knowledge haven’t been utilized before in such lower bounds. We will highlight these contributions earlier and more explicitly in the paper. ))}

Our lower bounds also use generalized reductions from differential privacy to reconstruction rates (Section~\ref{sec:preliminaries}), which may find more application in the edge weight setting.

\revise{We note that the assumption on $\delta < (1/n)^{\Omega(1)}$ in our lower bounds is common in the privacy literature because, for example a mechanism that releases a uniformly drawn $\delta$-fraction of the dataset would satisfy $(0, \delta)$-DP.
We refer to the discussion in \citet{Vadhan_2017}.
}
Nevertheless, it remains an interesting open question to show tight bounds for larger values of $\delta$.

\subsection{Related Work}

\paragraph{Graph optimization problems.}
%\mstLong is a classical graph optimization problem with early algorithms due to Jarník, Kruskal, and Prim~\citep{jarnik1930,kruskal_1956,Prim_1957}.
%Its applications range from network design to tree-structured graphical models and synthetic data generation~\citep{chow_1968,McKenna_Miklau_Sheldon_2021}.
%\mwpmLong is a canonical model for pairing and assignment problems, including bipartite settings where two sides of a market or network (e.g., doctors and hospitals, patients and kidneys) must be paired to minimize total cost.
%A foundational polynomial-time approach is Edmonds' blossom algorithm~\citep{Edmonds_1965}; for integer weights bounded by $N$, the scaling algorithm of~\citet{Duan_Pettie_Su_2018} solves the general-graph problem in $O(m\sqrt{n}\log(nN))$ time.
\mstLong~and \mwpmLong~are classic graph optimization problems, with applications including tree-structured graphical models, synthetic data generation, and market or network matching. 
\hcLong is an unsupervised learning technique that - unlike \emph{flat} clustering methods such as $k$-means - organizes data into clusters at multiple levels of granularity.
The problem dates back more than half a century~\citep{Ward1963} and has applications across many domains, including network analysis~\citep{Leskovec2020}, text analysis \citep{steinbach2000}, and biology~% and various subareas in biology 
\revise{\citep{jardine1968model,sneath2005numerical,diez2015novel,eisen1998cluster,sotiriou2003breast}}.
%\lukas{((TODO)) cost function + paper by } % sry, Comes later.

\paragraph{Reconstruction attacks.}
Our lower-bound proofs are part of a long line of reconstruction attacks in differential privacy.
The seminal work of~\citet{dinur2003revealing} showed that answering too many queries too accurately can allow reconstruction of most of a private database.
Subsequent work developed this connection into general lower-bound methods for private data release, including iterative constructions~\citep{Gupta2012}, lower bounds based on reconstruction and fingerprinting ideas~\citep{de2012lower}, and robust traceability arguments~\citep{dwork2015robust}.
%Many other techniques are now known to show lower bounds for differentially private algorithms, including packing lower bounds, fingerprinting techniques, communication complexity reductions, and discrepancy lower bounds.
%A modern exposition of such techniques appears, for example, in~\citet{Vadhan_2017}.
In the graph setting,~\citet{Sealfon_2016} and~\citet{eden2025triangle} achieve lower bounds via reconstruction attacks; we improve on the results of Sealfon, while the latter result applies to a different privacy model.

\paragraph{Private graph algorithms.}
%There is a broad literature on differential privacy for graph-structured data under various notions of neighboring databases, e.g., edge-level, node-level,  weighted, and dynamic (continual observation) notions of neighboring databases.
Early graph-DP work studied edge-level statistics such as degree distributions~\citep{Hay2009}, while node-DP was formalized and studied by~\citet{Kasiviswanathan2013}.
Under edge-level DP, \citet{Eli2020} gave near-optimal algorithms for approximating cuts via synthetic graphs, and \citet{Liu2024} obtained optimal bounds for several other cut-related private graph approximation tasks.
In a \emph{local} version of edge-DP where each node releases its own perturbed view of the graph, \citet{eden2025triangle} studied triangle counting and used reconstruction-style lower-bound techniques.
%Recent work has considered graph streams and dynamic graphs, including continual release lower bounds~\citep{aryanfard2025improved}, sublinear-space algorithms in the continual release model~\citep{epasto2025}, and fully dynamic algorithms for edge-DP graph databases~\citep{raskhodnikova2025}.
%Closest in spirit to our edge-weight lower bounds are works that encode hard instances into graph weights or graph structure and argue that an accurate output reveals too much information about the private input. RP: Need concrete references if this is to be included
Private matching in a different privacy model and objective has been studied in discrete allocation and graph models~\citep{justin,dinitz_2025}. These privacy models are distinct from the edge-weight model, and the techniques do not easily carry over.

\paragraph{Edge-weight private graph optimization.}
The edge-weight model was introduced by~\citet{Sealfon_2016}, modeling the situation where graph structure is public but weights are sensitive.
He studied differentially private release of approximate shortest paths and distances, and also gave algorithms and lower bounds for releasing \mstLong, \ssptLong, and \mwpmLong under the $\ell_1$ neighboring relation on weights.
A recent line of work has studied private release of \emph{all-pairs shortest-path distances} under the $\ell_1$ neighboring relation, with lower bounds shown using discrepancy-based methods~\citep{fan_2022,Chen2023,greg_2024}.
MST was studied under the $\ell_\infty$ neighboring relation in \citet{hladik_2024} and \citet{pagh_2025}, showing tight error bounds in pure and approximate DP settings, respectively. \revise{To our knowledge, \citet{hladik_2024} provide the only other topology-dependent lower bounds for edge-weight DP, showing the error of MST for the $\ell_1$ neighboring relation must grow with $\Omega(D)$, where $D$ is the diameter of the topology. Their argument uses the packing technique, which restricts it to pure DP ($\delta = 0$).}

\paragraph{Private hierarchical clustering.}
The non-private objective framework for hierarchical clustering was introduced by~\citet{dasgupta2016cost} and recently won a STOC Test of Time Award \citep{sigact_stoc_tot}.
It was further developed by~\citet{Cohenaddad2019} and the differentially private version was first studied by~\citet{pmlr-v202-imola23a}, who gave algorithms and proved a packing-based $\Omega(n^2/\epsilon)$ additive-error lower bound for pure edge-DP.
\citet{deng2025pricedifferentialprivacyhierarchical} subsequently studied hierarchical clustering in the edge-weight model and showed that hard instances can be embedded into the weights of the complete graph, implying an $\Omega(n^2/\epsilon)$ lower bound under pure DP.
Our hierarchical-clustering lower bound follows the edge-weight line of~\citet{deng2025pricedifferentialprivacyhierarchical}, but strengthens the picture by applying to approximate DP and to \revise{beyond the complete topology}.
\revise{They also derived an algorithm achieving $\tilde{O}(\frac{1}{\epsilon})$ multiplicative error for graphs in which every edge weight is at least $1$; we do not compare with this algorithm as this assumption is too specialized for our purposes.}

%% file: tables/results-l1.tex
%\begin{table}[]
%    \begin{tabular}{ccccc}
%                  & \textbf{$\ell_1$-Pure}             & \textbf{$\ell_{\infty}$-Pure}        & \textbf{$\ell_1$-Approximate}                & \textbf{$\ell_{\infty}$-Approximate}           \\
%    \textbf{MST}  & $\Theta(\frac{n}{\epsilon}\log n)$ & $\Theta(\frac{n^2}{\epsilon}\log n)$ & $\Theta(n \sqrt{\log n})$                    & $\tilde{\Theta}(n^{3/2})$                      \\
%    \textbf{SSPT} & $\Theta(\frac{n}{\epsilon})$       & $\Theta(\frac{n^3}{\epsilon})$       & $\Theta(\frac{n}{\epsilon})$                 & $\Theta(\frac{n^2}{\epsilon})$                 \\
%    \textbf{MWPM} & $\Theta(\frac{n}{\epsilon}\log n)$ & $\Theta(\frac{n^2}{\epsilon}\log n)$ & $\Omega(n)$ / $\mathcal{O}(n \sqrt{\log n})$ & $\Omega(n)$ / $\mathcal{O}(n^2 \sqrt{\log n})$
%    \end{tabular}
%    \caption{Landscape for the problems of privately releasing \mstLong~({\bf MST}), Single-Source~\ssptLong ({\bf SSPT}) and \mwpmLong~({\bf MWPM}).
%    The bounds for $(\epsilon, \delta)-$DP holds for $\delta \leq 1/n$.}
%\end{table}

% \usepackage{multirow}
%TODO: replace cline with cdashedline
\begin{figure}[t]
\def\arraystretch{1.2}%  1 is the default, change whatever you need
\centering
\begin{tabular}{c|clp{3.6cm}c}
    {\bf Problem} & {\bf Error Bound} & {\bf Reference} & {\bf Graph type} & {\bf Parameters}\\
    \hline
    \multirow{5}{*}{MST} & $\mathcal{O}(n \cdot \log n)$ & \cite{Sealfon_2016} & Any & $\delta = 0$\\
    & $\mathcal{O}(n \cdot \log(m/n))$ & \cite{aamand-personal-communication} & Any & $\delta = 0$\\
    & $\Omega(n)$ & \cite{Sealfon_2016} & Multi-edge star graph & $\delta < c_\varepsilon$\\
    &$\Omega(D)$ & \cite{hladik_2024} & Diameter $D$ & $\delta = 0$ \\
    & $\Omega(n \cdot \log(m/n))$ & \textcolor{blue}{\Cref{th:lb-mst}} & Worst-case & $\delta < (n/m)^{\Omega(1)}$\\
    & $\Omega(n)$  & \textcolor{blue}{\Cref{coro:mst-lb}} & Min-cut $\geq 5 \log(n)$ & $\delta < \frac{c_\epsilon}{m}$\\
    \hline
    \multirow{4}{*}{MWPM} & $\mathcal{O}(n \cdot \log n)$ & \cite{Sealfon_2016} & Any & $\delta = 0$\\
    & $\mathcal{O}(n \cdot \log(m/n))$ & \cite{aamand-personal-communication} & Any & $\delta = 0$\\
    & $\Omega(n)$ & \cite{Sealfon_2016} & Collection of 4-cycles & $\delta < c_\varepsilon$\\
    & $\Omega(n \cdot \log(m/n))$ & \textcolor{blue}{\Cref{th:lb-matching}} & Worst-case & $\delta < (n/m)^{\Omega(1)}$\\
    \hline
    \multirow{3}{*}{HC} & $\mathcal{O}(n^2 \log n)$ & \cite{pmlr-v202-imola23a} & Any & $\delta = 0$\\
    & $\Omega(n^2)$ & \cite{deng2025pricedifferentialprivacyhierarchical} & Complete graph & $\delta = 0$\\
    & $\Omega(n^2)$ & \textcolor{blue}{\Cref{thm:hc-reconstruction}} & Degree-$\cO(1)$ expander & $\delta < \frac{c_\epsilon}{m}$\\
    %& $\Omega(\frac{\phi(G)}{d_{max}(G)}n)$ & Theorem~\ref{thm:hc-reconstruction} & Any topology with $d_{max}(G) \geq 8, \phi(G) \geq 8n$.
    \hline
    \end{tabular}
%\begin{tabular}{l|cll}
%    \multicolumn{1}{c|}{\textbf{Problem}}        & \multicolumn{1}{c}{\textbf{Error Upper Bound}}                           & \multicolumn{1}{c}{\textbf{Error Lower Bound}}                                       \\ \hline
%    \textbf{MST}   & $\mathcal{O}(n \cdot \log(m/n))\,$\scc{A} & \allbold{$\Omega(n \cdot  \log(m/n)\,)$}\scc{\cref{th:lb-mst}}      \\ 
%    \textbf{MWPM} &  $\mathcal{O}(n \cdot \log(m/n)\,)$\scc{A} & \allbold{$\Omega(n \cdot \log(m/n)\,)$}\scc{\cref{th:lb-matching}} \\ 
%    \textbf{HC}   & \allbold{$\mathcal{O}(n^2)$}\scc{\cref{thm:apsp}}                          & $\Omega(n^2\cdot\log n)$\scc{A}                                                                   \\ \hline
%\end{tabular}
\caption{Upper and lower bounds for the error of differentially private graph optimization problems: MST, MWPM, and HC under the $\ell_1$ neighboring relationship \revise{parameterized by the number of vertices $n$ and edges $m$.}
Upper bounds hold for pure differential privacy ($\delta = 0$), while lower bounds are for $(\epsilon, \delta)$-DP.
%\lukas{((NOTE) L) I guess this is only impressive if the upper bound for HC also matches. But I guess there is an exponential time alhorithm due to Jacob that achieves it up to poylog factos?}
All bounds omit a multiplicative factor $1/\varepsilon$ for readability; $c_\varepsilon > 0$ is a constant that depends only on $\varepsilon$.}
\label{tab:results}
\end{figure}
% All upper lower bounds 1/eps
% weights as sum of vvecotrs?
% Shortest paths on mutual infomraiton
% ADD THESE referens in some text.\textbf{B)} \cite{hladik_2024}}

%\lukas{TODO: Add fancy errors showing implication. And add missing boubdn matching. Dependency on delta for C)}

%% file: content/preliminaries.tex
\section{Preliminaries}\label{sec:preliminaries}
For two vectors $\W, \W' \in \R^m$, we define the \emph{Hamming distance} $d_H(\W, \W') := \sum_{i=1}^{m} \mathbb{1}[w_i \neq w_i']$ to be the number of coordinates in which they differ.
For sets $S, S'$ we extend the definition by the symmetric difference $d_H(S, S') := |S \setminus S'| + |S' \setminus S|$.
We write $\x \sim_H \x'$ if $d_H(\x,\x')\leq 1$.

\paragraph{Graph Terminology} We consider a finite, simple and undirected graph $G = (V, E, \W)$, where the set of $n$ vertices $V$ and the set of $m$ edges $E$ are public and the weight vector $\W = (w_1, \cdots, w_m) \in \R^m$ is private.
We use $\vec{w} \in \R^m$ and $w : E \to \R$ interchangeably to denote the edge weights.
%For two vertices, we shorten $w(\{u, v\})$ by writing $w(u, v)$.
We refer to $F = (V, E)$ as the \emph{topology} of $G$.
Throughout the paper, we assume $m > n$ and that $F$ is connected.
We denote by $\cG$ (respectively $\cG_\omega$) the family of all unweighted (weighted) graphs.

\paragraph{Weights of Sets and Cuts} For a subset of edges $S \subseteq E$ of edges, we denote the \emph{weight} as $w(S) = \sum_{e\in S}w_e$. We refer to a partition $A,B$ of $V$ as a \emph{cut} of $G$. The \emph{weight} of the cut is given by $w(A,B) := \sum_{uv\in E \cap (A \times B)} w_{uv}$. 
Similarly, for any edge set $E' \subseteq E$, we denote $E'(A,B) = |\{(u,v) \in E':u \in A,\, v \in B\}|$ to be the size of the cut, or the number of edges crossing between $A,B$. We refer to any cut $A,B$ satisfying $\frac{1}{3}n \leq |A|, |B| \leq \frac{2}{3}n$ as a \emph{balanced cut}.

\paragraph{\revise{Spanning} Trees and Perfect Matchings} A \emph{spanning tree} is an acyclic subset $T\subseteq E$ of size $n-1$ making the graph connected.
A \emph{matching} $M \subseteq E$ is a set of pairwise vertex-disjoint edges, i.e., $e \cap e' = \emptyset$ for any two distinct $e, e' \in M$. 
For a matching $M$ and a subset $S \subseteq V$, we let $S(M)$ denote those vertices in $S$ that are incident to an edge of $M$.
A \emph{perfect matching} is a matching that covers every vertex of $V$.
We denote the set of all spanning trees and perfect matchings of $G$ as $\cS(G)$ and $\cP(G)$, respectively.
A \emph{minimum spanning tree} (MST) is a spanning tree $T^*$ that minimizes $w(T)$ among all spanning trees $T \in \cS(G)$; analogously, a \emph{minimum weight perfect matching} (MWPM) is a perfect matching $M^* \in \cP(G)$ that minimizes $w(M)$ among all $M \in \cP(G)$.

\paragraph{Hierarchical Clustering} 
A \emph{hierarchical clustering} (HC) for a graph $G$ is represented by a rooted tree $T$
whose leaves correspond to $V$, and where each internal node corresponds to the union of the leaves in its subtree.
%and a merge of two subtrees combines the vertices contained in the subtrees until the root contains all of $V$.
Traversing the tree from the root to the leaves corresponds to successively refining the clustering.
%partitioning the vertices of $G$ into refined clusters.
In seminal work~\citep{dasgupta2016cost}, Dasgupta introduced a cost function $\hccost{w}(T)$ evaluating the quality of a hierarchical clustering, setting the stage for optimization. 
Intuitively, it charges the total weight separated at each merge times the size of the subtree at that merge, encouraging less weight to cross higher merges. Formally,

\begin{definition}[\cite{dasgupta2016cost} Dasgupta's cost objective]
    Given an input graph $G = (V, E)$ together with weights $w:E\rightarrow \R_{\geq0}$ and a hierarchical clustering tree $T$, the Dasgupta cost function $\hccost{w}$ is given by 
    \begin{align}
    \hccost{w}(T) := \sum_{\{u,v\}\in E} \abs{T[u \vee v]}\cdot  w(\{u, v\})\,,
    \end{align}
\end{definition}
where $|T[u \vee v]|$ denotes the number of leaves in the subtree induced by the lowest common ancestor of $u$ and $v$ in $T$. 
As noted by Dasgupta, we may assume without loss of generality that $T$ is a rooted \emph{binary} tree, because every non-binary hierarchical clustering tree can be transformed into a binary one without increasing $\hccost{w}(T)$.
\ifdefined\PODS \else
We illustrate this cost function on a small example in Appendix~\ref{app:details}.
\fi
\paragraph{\textbf{Differential Privacy}}
%Differential privacy (DP) \cite{Dwork2006} was developed as a tool to hide the influence of any single individual on the output of an algorithm. In the context of graphs, several different models have been proposed, for example \emph{edge}-level \cite{Hay2009} and \emph{node}-level \cite{Kasiviswanathan2013}. We refer to the survey \cite{Li2023} for a good survey for DP on graph data.
In this work, we consider \emph{edge-weight} differential privacy \cite{Sealfon_2016}, where the private information is encoded in the weights themselves. 
%We say that two vectors and $\vec{w},\vec{w'}$ are $\ell_1$-neighboring (denoted $w \sim_1 w'$), if 
We say that two graphs $G = (V, E, \vec{w})$ and $G' =(V', E', \vec{w}')$ are neighboring (denoted $G \sim G'$) if $V = V'$, $E = E'$ and $\norm{\vec{w}-\vec{w}'}_1\leq 1$, i.e., their $\ell_1$ distance is at most one.
%,\sim G'$ are neighboring, if their weight functions $\vec{w}$ 
%In this work, we define the neighboring relationship in terms of the $\ell_1$ distance if the weight vector:
\begin{definition}[\cite{Dwork2006} $(\epsilon, \delta)$-Private Algorithm]\label[definition]{def:ew-dp}
    Let $\epsilon \geq 0$ and $\delta \in [0,1]$.
    A mechanism $\cA:\cG_\omega \rightarrow \cY$ is $(\epsilon, \delta)$-(edge-weight)-DP, if for every pair of neighboring graphs $G \sim G'$, and all measurable sets of outputs $Y \subseteq \cY$, we have,
    \begin{align}
        \Pr{\mathcal{A}(G) \in Y} \leq e^{\epsilon} \Pr{\mathcal{A}(G') \in  Y} + \delta.
    \end{align}
\end{definition}
Importantly, $F = (V,E)$ is a \emph{public} topology, and privacy is with respect to the weights. 
%This contrasts with other privacy notions include \emph{edge}-DP \cite{Hay2009} and \emph{node}-DP \cite{Kasiviswanathan2013} which focus on the protection of the particular membership of a single edge or vertex in the graph respectively. The best notion to use depends on the application; as mentioned earlier, edge-weight DP is applicable when $F$ is a known network and sensitive behavior affects the weights. %% Lukas: Covered in Private graph alg section

\paragraph{Reconstruction Attacks.} Our DP lower bounds will be shown via reconstruction attacks, which are algorithms that try to reconstruct the private dataset $\x$ based on the algorithm output. If the attack reconstructs coordinates of $\x$, they demonstrate an impossibility for DP, and any reasonable notion of privacy (see \citet{Vadhan_2017} for an overview of reconstruction attacks).

We will use two types of reconstruction attacks in our lower bounds.
In the first type, we consider an algorithm that attempts to reconstruct a fixed coordinate $i$ of $\x$, and show that any DP algorithm cannot be too successful at doing so. 
This is a slight generalization of \citet[Lemma 5.3]{Sealfon_2016} to general data domains $\cX$.

\begin{restatable}%[{\cite{Vadhan_2017}}]
{lemma}{reident}\label[lemma]{lem:reidentification}%\rasmusnote{It is unclear why we credit Vadhan for this and not Sealfon so I removed the cite}
Let $\cB:\cX^d \rightarrow \cX^d$ be any mechanism that is  $(\epsilon, \delta)$-differentially private under the Hamming neighborhood relationship, \revise{i.e. $d_H(\x, \x')\leq 1$ for neighboring datasets $\x \sim \x'$}.
Then a dataset $\X \leftarrow  \cX^d$, drawn uniformly at random, we have  for each $i \in [d]$, we have
$\Pr{\cB(\X)_i = x_i}\leq\tfrac{\exp(\epsilon)}{|\cX|} + \delta$. %&   & \Ex{\cM(\vec{x})_i = x_i} \leq d\cdot \PAREN{\frac{\exp(\epsilon)}{|\cX|} + \delta}\,,
\end{restatable}
As we will see, using a data domain of size $|\cX| = \frac{m}{n}$ will be critical to obtaining tight lower bounds.
We also consider reconstructing an entire set $I \subseteq [d]$ of indices of a binary dataset $\x$. 
An attack is then a mechanism $\cB : \{0,1\}^d \rightarrow 2^{[d]} \times \{0,1\}^d$, and its success is measured by
\begin{equation}\label{eq:succ}
    \suc(\cB) = \frac{1}{2} - \frac{\Ex{d_{H}(\vec{x}\vert_I, \vec{y}\vert_I)}}{\Ex{|I|}},
\end{equation}
where $\x \sim \{0,1\}^d$ and $(I, \vec{y}) := \cB(\x)$, and $\x|_I$ indicate the subvector of $\x$ on indices $I$. 
Intuitively, success is the normalized error that randomly guessing $I$ would yield (namely $\frac{1}{2})$ minus the normalized error made by $\cB$ in the set $I$. 
DP limits the possible success as follows.

\begin{restatable}{theorem}{reconst}\label{lem:reconst}
    Let $\cB:\{0,1\}^d\rightarrow 2^{[d]} \times \{0,1\}^d$ be any $(\epsilon, \delta)$-DP algorithm and let ${\x \sim \{0,1\}^d}$ be drawn uniformly at random. Then, $\suc(\cB) \leq \frac{e^\epsilon-1}{2(e^\epsilon+1)} + \frac{d\delta}{e^\epsilon+1}$.
\end{restatable}
The proof appears in \cref{proof:subsetrec}. It differs from prior reconstruction attacks because $I$ is itself a differentially private output. 
When $\delta = o(\frac{1}{d})$, the second term becomes negligible, and when $\epsilon < 1$, we have that $\suc(\cB) =  O(\epsilon)$.

% Another proof for this. In step to any valid $\neq j$ would work.
% \begin{proof}[Proof sketch for B outputting single value]
% \begin{align}
%     \Pr[B(X) = X_i] &= \sum_{j\in [n]} \Pr[B(X) = X_i | X_i = j]\cdot \Pr[X_i = j] \\
%     &\leq\sum_{j\in [n]}\left(\dfrac{e^\epsilon}{n}\Pr[B(X) = j \mid X_i = j] +\delta\right) \cdot \Pr[X_i = j] \\
%     &=\delta+ \dfrac{e^\epsilon}{n}\sum_{j\in [n]}\left(1 - \sum_{j' \in [n]\setminus j} \Pr[B(X) = j' \mid X_i = j] \right) \cdot \Pr[X_i = j] \\
%     &=\delta+ \dfrac{e^\epsilon}{n}\sum_{j\in [n]}\left(1 - \sum_{j' \in [n]} \Pr\left[B(X) = j' \mid X_i = j\right] +\Pr\left[B(X) = j\mid X_i  = j\right]\right) \cdot \Pr[X_i = j] \\
%     &=\delta+ \dfrac{e^\epsilon}{n}\sum_{j\in [n]} \Pr\left[B(X) = j\mid X_i  = j\right] \cdot \Pr[X_i = j] \\
%     &\leq \delta+ \dfrac{e^\epsilon}{n} \cdot 1 \\

%     \end{align}
% \end{proof}

%% file: content/mst.tex
\input{content/worst-case-bounds/mst}

\subsection{Lower Bounds for General Topology Classes}\label{sec:rc-mst}
 We will design a subset reconstruction attack that uses any given approxmate DP MST algorithm \mstalg. 
 Assuming the additive cost of \mstalg{} is sufficiently low, the reconstruction success will be high enough to form a contradiction with \cref{lem:reconst} and establish an error lower bound for any DP algorithm.

Our reconstruction attack works as follows: 
the dataset is drawn from $\x \sim \operatorname{Uni}(\{0,1\}^m)$, and then arbitrarily encoded into $E$ with a labeling function $\ell : E \rightarrow [m]$. 
The weights are set as $w(e) = R x_{\ell(e)}$, where $R$ is a factor used to amplify the weight that we will choose later. 
Then, $\mstalg$ is run on $(V, E, \vec{w})$ to produce a tree $T \subseteq E$.  %(we assume $(V,E)$ is connected). 
The returned tree $T$ contains edges whose weight is more likely to be $0$, and thus the reconstructed set is set to be $I = \{\ell(e) : e \in T\}$, and all guesses $\{y_i : i \in I\}$ are $y_i = 0$. This procedure appears in Algorithm~\ref{alg:ra-mst}.

\begin{algorithm}[t]
\caption{Subset Reconstruction Attack $\ra^{\mathsf{MST}}$ for MST}\label{alg:ra-mst}
    \begin{algorithmic}[1]
    \REQUIRE{Public graph topology $F = (V, E)$, dataset $\x \in \{0,1\}^m$, parameter $R > 0$, algorithm $\mstalg$.}
        \STATE{\ifdefined\PODS\vspace{-1em}\fi Initialize weights $\vec{w} \in \left\{0, R\right\}^m$.}
        \STATE{Let $\ell : E \rightarrow [m]$ be any bijection from $E$ to $[m]$ labels.}
        \STATE Set $w(e) := Rx_{\ell(e)}$ for all $e \in E$.
        \COMMENT{Encode the dataset into the weights}
        \STATE{Let $G = (F, \vec{w})$.}
        \STATE{Compute minimum spanning tree $T := \mstalg(G)$.}
        \STATE{Initialize reconstruction set $I = \{\ell(e) : e \in T\} \subseteq [m]$.}
        \STATE{\textbf{for each} $i \in I$, guess $y_{i} = 0$.}
        \RETURN{Reconstruction set $I$, reconstructed vector $\vec{y} = \{y_i : i \in I\}$.}
    \end{algorithmic}
\end{algorithm}

We will show that if $\mstalg$ has sufficiently low additive error, then this attack has a success rate much higher than $0$. Specifically, we will show $\EX[\sum_{i \in I} x_i] \leq 0.1 \EX[|I|]$. By construction, $w(T) = R \sum_{i \in I} x_i$. Using the additive error guarantee of $\mstalg$, we have $\EX[\sum_{i \in I} x_i] = \frac{1}{R}\EX[w(T)] \leq \frac{1}{R} w(T^*) + \frac{1}{10}n$. Thus, we need to show that the \emph{optimal} MST in $(V,E,\vec{w})$ is likely to have low weight. This is where we assume that the \emph{minimum cut} of $F$ is not too small. Recall that the minimum cut is given by $\min_{A \sqcup B = V} E(A,B)$. We use a concentration inequality, and a union bound over all cuts, to argue that if the mincut of $F$ is at least $5 \log n$, then with high probability all cuts have an edge labeled $0$. This immediately implies that an MST of weight $0$ exists. In order for the union bound to work, we use a result of~\citet{karger1993global} bounding the number of approximate minimum cuts, as these cuts have higher failure probability. Formally,
\begin{lemma}\label[lemma]{lem:mst-add-err}
    Suppose the graph topology $F = (V,E)$ has minimum cut $\lambda \geq 5 \log(n)$, and the random dataset $\vec{x} \sim \operatorname{Uni}(\{0,1\}^m)$ is embedded into $\vec{w}$ via an arbitrary bijection. 
    Then, with probability at least $1-\frac{3}{n^2}$, there exists a zero-cost spanning tree $T \subseteq E$, i.e. $w(T) = 0$.
\end{lemma}
The proof appears in \Cref{app:mst-add-err}. Note that $\EX[|I|] = |I| = n-1$ always. Having established bounds on both $\EX[\sum_{i \in I} x_i]$ and $\EX[|I|]$, we can show a reconstruction attack lower bound:
\begin{theorem}\label{thm:mst-lb}
    Let $(V,E)$ be a public graph topology with minimum cut at least $5 \log(n)$. Suppose that $\mstalg$ attains expected additive error $\leq \frac{R}{10} n$ for some $R > 0$. 
    Then $\suc(\ra^{\mathsf{MST}}) \geq 0.4$.
\end{theorem}
\begin{proof}
    Let $\cE$ denote the event that the optimal MST has weight that is not $0$. We have that 
    \begin{align*}
        \EX\left[\sum_{i \in I} \mathbb{1}[x_i \neq y_i]\right] = \EX\left[\sum_{i \in I} x_i\right] &\leq \EX\left[\sum_{i \in I} x_i\middle \vert \neg \mathcal{E}\right] + \Pr{\mathcal{E}} \EX\left[\sum_{i \in I} x_i\middle \vert \mathcal{E}\right] \\
        &\leq \EX\left[\sum_{i \in I} x_i\middle \vert \neg \mathcal{E}\right] + \Pr{\mathcal{E}} \cdot n\cdot R.
    \end{align*}
    By Lemma~\ref{lem:mst-add-err}, we know $\Pr{\mathcal{E}} \leq \frac{3}{n^2}$. Furthermore, by the construction of $I$, we know $R\sum_{i \in I} x_i = w(T)$. 
    Given that $\neg \mathcal{E}$ occurs, we have $\EX[w(T)] \leq \frac{R}{10}n$ by the error guarantee of $\mstalg$. Thus, the above bound is $\frac{n}{10} + \cO\big(\frac{R}{n}\big)$.
\end{proof}
Theorem~\ref{thm:mst-lb} and Lemma~\ref{lem:reconst} immediately imply the following corollary: 
\begin{corollary}\label[corollary]{coro:mst-lb}
    For any topology $F = (V,E)$ with minimum cut at least $5 \log(n)$, there is no $(\epsilon, \delta)$-DP algorithm $\mstalg$ attaining $\frac{1}{10\epsilon}n$ additive error for any $\epsilon \leq 1, \delta \leq \frac{0.01\epsilon}{m}$.
\end{corollary}
\begin{proof}
Suppose the contrary, and instantiate $\ra^{\mathsf{MST}}$ with $\mstalg$ and $R = \frac{1}{\epsilon}$. By group privacy, $\ra^{\mathsf{MST}}$ satisfies $(\epsilon \lceil R \rceil, \delta \lceil R \rceil e^{\epsilon \lceil R \rceil})$-DP, which means it satisfies $(2, \frac{10}{\epsilon}\delta)$-DP as $\epsilon \lceil R \rceil < 2$.
By Lemma~\ref{lem:reconst}, this means it can only be a $\frac{1}{2} - \frac{1}{1+e^2} + m \delta < 0.39$-successful reconstruction attack, yet this contradicts Theorem~\ref{thm:mst-lb}.
\end{proof}
While looser than Theorem~\ref{th:lb-mst} by a $\log(\frac{m}{n})$ factor, this result establishes that the error of MST must still grow linearly under $(\epsilon, \delta)$-DP for many topologies.
We note that this proof would still work if $F$ had \emph{any} linear-sized subgraph with large enough min-cut.

%% file: content/worst-case-bounds/mst.tex
\newcommand{\Rval}{\ensuremath{\min\left(\frac{\ln n}{2 c \epsilon}, \frac{\ln(1/\delta)}{2 c \epsilon}\right)}}

\section{Private Minimum-Weight Spanning Trees}\label{ch:enc-mst}
In Section~\ref{sec:mst-wc}, we show that there is a worst-case topology $F$ on which the additive error of MST under DP is tight in $\Theta\left((n/\epsilon) \cdot \log (m/n)\right)$ (for small enough $\delta$) achieved by the Laplace mechanism.
Going beyond worst-case graphs, in \cref{sec:rc-mst}, we show that the lower bound of $\Omega(n/\epsilon)$ due to \cite{Sealfon_2016} holds on all graph topologies with minimum cut $\Omega(\log(n)$.

\subsection{Tight Lower Bounds on Worst-Case Topologies}\label{sec:mst-wc}
\begin{figure}[t]
    \centering
    \scalebox{1.2}{
        \input{fig/mst-reduction.tikz}
    }
    \caption{{\bf a)} Encoding the vector $\X = (x_1, \cdots, x_n)$ into the \emph{minimum-weight spanning tree}.
     {\bf b)} The encoding for \emph{min-weight perfect matching} for a vector of length $\lceil \alpha n\rceil$.
     We pad the left side with the gray vertices.
    In both cases, each vertex $r_j$ represents the integer $j$, and we set the weight of the edge from $l_i$ to $r_j$ to $0$, if $x_i = j$ and otherwise to $R = \mathcal{O}(\frac{1}{\epsilon}\ln n)$.
     }\label{fig:encode-as-graph}
\end{figure}

\newcommand{\deltaless}{\ensuremath{\frac{1}{n}}}

Assuming \revise{$\delta \leq (n/m)^{\Omega(1)}$}, we now close the gaps between the known upper and the lower bounds.
The idea is to encode a random vector $\X \in [m/n]^{d}$ into the MST of a dense graph $G$ so that releasing an overly accurate MST by a differentially private mechanism would contradict the reconstruction rate in \cref{lem:reidentification}. 
As a by-product, we also get a lower bound for $\epsilon$-DP, obtained by taking the limit of $\delta$ towards $0$ recovering the packing-based lower bound obtained by \citet{hladik_2024}.

\begin{theorem}[Worst-Case Lower Bound MST]\label{th:lb-mst}
Let $\epsilon >0$ and $\delta \leq (n/m)^{\Omega(1)}$ where $m>2n$.
There exists a graph topology $F=(V,E)$ and a distribution $\cD_\omega$ of weights such that for any $(\epsilon, \delta)$-DP algorithm $\cA$ that outputs an approximate MST $T$
under the $\ell_1$ neighboring relationship, if the weights $\W\sim \cD_\omega$ then the expected additive error satisfies 
\[ 
 w(T) - w(T^*) \geq \Omega\big((n/\epsilon) \cdot \min(\ln(1/\delta), \ln(m/n))\big)\,,
\] where $T^*$ is the optimal MST.
\end{theorem}
\noindent \revise{We will start with the case $m=n^2$, and focus on the general case later. The full encoding scheme appears in  Figure~\ref{fig:encode-as-graph}.}
Throughout the construction, we assume the topology $F \in\cG$ to be public.

\paragraph{Encoding} Let $\Aenc_R:[n]^n \rightarrow \cG_\omega$ be the encoding function that takes some dataset  $\vec{X} \in [n]^n$ and returns a weighted graph that encodes the dataset into the MST.
The parameter $R>1$ is used to control the edge weights.
We construct a new connected graph $G = (V_1 \cup V_2, E, \vec{w})$ with $2n$ vertices and $n^2+(n-1)$ edges in the following way: 
Let $V_1 = \{\ell_1, \cdots \ell_n\}$ and $V_2= \{r_1, \cdots, r_n\}$ be two vertex sets each of size $n$.
Each vertex $\ell_i \in V_1$ represents one coordinate of $\vec{x}$ and has edges to each vertex in $V_2$.
Set $w\big(\{\ell_i, r_{x_i}\}\big) := 0$ and $w\big(\{\ell_i, r_j\}\big) := R$ for $j \neq x_i$.
%together with new edges $e_i^1, \cdots, e_i^n$ from $v_i$ to all $r_j \in V_1$. 
%represents one coordinate that points to a corresponding representative in $V_2$.
%Therefore, for each $x_i$ of $\X$, add a new vertex ${v_i}$ 
% 5   Set the weights for each $j \in [n]$ to $w(e_i^j) = 0$ if $j = X_i$ and $w(e_i^j) = R$ otherwise.
We ensure connectivity by adding a path through the vertices in $V_2$ with zero weighted edges (contributing the $n-1$ extra edges).

\paragraph{Decoding} Let ${\Adec_G: \cS(G) \rightarrow [n]^n}$ be a decoder and assume that the topology $G$ was produced by the encoder (to simplify notation we assume $F$ is public as soon as the instance has been created).

Given a spanning tree $T\in \cS(G)$, we reconstruct the vector $\vec{y}\in [n]^n$ as follows: for each $i \in [n]$, choose any $r_j \in V_2$ with $\{\ell_i, r_j\} \in T$ and set $y_i := j$.
%$i \in [n]$, any $r_j \in V_2$ such that $\{\ell_i, r_j\} \in T$ and then set $y_i =j$ to the encoded value.
Note that any vertex in $V_1$ must have at least one edge in $T$ by construction, and therefore, the reconstruction for each coordinate is well-defined.
In case there are multiple, we break ties arbitrarily.

%that takes a weighted graph $G \in \cG_\omega$ with vertices $V = \{v_1, \cdots, v_d, r_1, \cdots, r_n\}$, together with some spanning tree $T \in \cS(G)$. \lukasnote{Check the following whether the indexing can be improved. }
%This function decodes the original vector $\vec{X} = (X_1, \cdots, X_d)$.

%\noindent Clearly, for each $v\in \{v_1, \cdots, v_d\}$ there is at least one edge in $\inc(v_i) \in T$ and thus, each $X_i$ of the decoded vector is defined.

Now define the attack $\Rec:[n]^n \rightarrow [n]^n$ as $\Rec_{\mstalg,R}(\x) :=
  \Adec_G\big(\mstalg(\Aenc_R(\x))\big)$ where $\mstalg$ is any (possibly randomized) MST algorithm.
By the construction above, it is clear that an optimal MST algorithm $\cA_{\operatorname{OPT}}$ recovers all coordinates perfectly, i.e. $\Rec_{\cA_{\operatorname{OPT}},R}(\vec{x}) = \vec{x}$.
This leads to the following two observations.

\begin{observation}\label[observation]{obs:accuracy}
Let $R>0$ and $G = (F, \vec{w})$ be obtained from $\Aenc_R(\x)$.
If $T \in \cS(F)$ is any and $T^*$ the optimal spanning tree of $G$, then
\begin{align*}
  d_H\PAREN{\Adec_F(T),\x} \le \frac{w(T)-w(T^*)}{R} = \dfrac{w(T)}{R},
  \end{align*}
\end{observation}
Indeed, $w(T^*)=0$, and whenever $\Adec_F(T)_i \neq x_i$, the tree $T$ contains an edge $\{l_i,r_j\}$ with $(j\neq x_i)$, which has weight $R$. These bad edges are distinct for different $i$, so the number of decoding errors is at most $w(T)/R$.

%Given any MST algorithm $\cA$ and vector $\vec{X} \in [n]^d$, we shortly denote $\Aencdec_R(\cA, \X) := \Adec(\cA(\Aenc_R(\X))$.
%We now establish a connection between two vectors that are close in Hamming distance to the $\ell_1$ sensitivity in  the encoded graph. 

\begin{observation}[Induced $\ell_1$]\label[observation]{lem:induced-hamming-neighborhood}
For every $\vec{X}\sim_H \vec{X'}$ and $R>0$, setting $k := 2\lceil R\rceil$, there is a chain of $\ell_1$-neighboring graphs, i.e. $\Aenc_R(\vec{X}):= G_0 \sim_1 G_{1} \sim_1 \cdots \sim_1  G_{k} =: \Aenc_R(\X')$.
\end{observation}
Each consecutive $G_i$ is obtained by successively changing a single weight by one. By group privacy~(see e.g. \cite[Lemma 7.2.2]{Vadhan_2017}), this means there is a $(2\lceil R \rceil \epsilon, 2\lceil R \rceil e^{2\lceil R \rceil \epsilon}\delta)$-DP guarantee on $\enc_R(\x)$ under the \emph{Hamming} neighboring relationship.
We are now ready to prove \cref{th:lb-mst}.
\begin{proof}[Proof of \cref{th:lb-mst}]
% Assume $\delta \leq 1/n$ 
Suppose for contradiction, that there exists an $(\epsilon, \delta)$-DP MST algorithm $\mstalg$ with expected additive error less than $\frac{n}{100}\cdot R$.
Assume a random dataset $\X \sim \operatorname{Uni}([n]^n)$.
Note that the decoding step is simply post-processing, and hence $\Rec_{\mstalg, R}$ parameterized by $\mstalg$ and $R$ is $(2\lceil R\rceil \epsilon, 2\lceil R\rceil e^{2\lceil R\rceil \epsilon}\delta)$-DP.
Now set $R := \Rval$ for some $c>1$.
Combining with \cref{obs:accuracy} and by the utility guarantee of $\mstalg$, we would leak at least $\frac{99}{100}\cdot n$ of the input coordinates in expectation.
Let $\vec{y} := \Rec_{\mstalg, R}(\vec{x})$ and by \cref{lem:reidentification},
\begin{align*}
\Pr{\vec{y}_i = \vec{x}_i} & \leq \frac{e^{2\lceil R\rceil \epsilon}}{n} + 2\lceil R\rceil e^{2\lceil R\rceil \epsilon}\delta\\
& \leq \frac{1}{n} \exp\left(\ln(n^{1/c})\right) + \frac{\ln(n)}{\epsilon c} \exp\left(\ln n^{1/c}\right) \delta \\
& = n^{(1-c)/c} + \frac{\ln(n)}{\epsilon c} n^{1/c} \delta\,.
%& \leq n^{(1-c)/c} + \frac{\log(n)}{\epsilon c} n^{1/c} \\
\end{align*}
By our assumption on $\delta = n^{-\Omega(1)}$, the second term vanishes.
Linearity of expectation gives
\begin{align*}
\Ex{\sum_{i\in [n]}\PAREN{\mathbb{1}\left[y_i = x_i\right]}} \leq  n \cdot n^{(1-c)/c} = n^{1/c}\revise{\,.}
\end{align*}
In particular, for a sufficiently large constant $c$, we contradict the utility guarantee of $\mstalg$.
%we have $\Pr{\Rec_{\cA, R}(\X)_i =
%X_i} < 0.99$ and by linearity of expectation, we contradict the utility guarantee of $\cA$.
Note also that we require $R \leq \frac{\ln n}{2c\epsilon}$ to
get any meaningful probability.
Thus, we have shown a lower bound of $\Omega((n/\epsilon) \ln n)$ in the dense case ($m = \Theta(n^2)$), for sufficiently small $\delta$.

It remains to show that this technique can be parameterized by the number of edges $m$ to yield a bound of $\Omega((n/\epsilon) \cdot \ln(m/n))$ for any $2n \leq m \leq n^2$.
Suppose, we want a lower bound for graphs with $m$ edges and $2n$ vertices. 
Without loss of generality, we treat $m/n$ as an integer.
We can build $n^2/m$ many subgraphs on $2(m/n)$ vertices (constructed as described above for the dense case with $n$ replaced by $m/n$).
 To make the full graph connected, we add a path through one representative vertex from each subgraph, adding only $\cO(n^2/m)$ low-cost edges that would only affect the bound by a constant.
Ignoring this lower-order term, observe that the number of edges is exactly $(m/n)^2 \cdot (n^2/m) = m$ (we ignore edges on the path of length $\leq n$).

For $i \in n^2/m$, let each subgraph encode an i.i.d. random vector $\vec{x}^{(i)}\in [m/n]^{m/n}$ of length $m/n$.
Since an $(\epsilon, \delta)$-DP algorithm satisfies the same privacy guarantees on the weights of each individual subgraph, the previous lower bound applied to a single subgraph states that we cannot reconstruct a single $\vec{x}^{(k)}$ with expected error better than $\Omega(\frac{m}{n \epsilon} \cdot \ln\frac{m}{n})$.
By linearity of expectation, we get a lower bound on the expected total error of $\Omega(\frac{m}{n\epsilon}\cdot \ln\big(\frac{m}{n})\cdot \frac{n^2}{m}\big) = \Omega\big(\frac{n}{\epsilon} \cdot \ln(\frac{m}{n})\big)$.
\end{proof}

%% file: fig/mst-reduction.tikz
\begin{tikzpicture}
	\begin{pgfonlayer}{nodelayer}
		\node [style=none] (58) at (-2.75, 1.25) {};
		\node [style=none] (59) at (-2.75, 0.25) {};
		\node [style=none] (60) at (1.75, -0.75) {};
		\node [style=none] (61) at (1.75, 3.25) {};
		\node [style=BLACK] (0) at (-1.5, 3) {};
		\node [style=GRAYN] (6) at (-1.5, -0.5) {};
		\node [style=BLACK] (7) at (-1.5, 0.75) {};
		\node [style=none] (9) at (-2.25, 0.75) {\scriptsize $l_i$};
		\node [style=none] (12) at (-2.25, 3) {\scriptsize $l_1$};
		\node [style=none] (13) at (-2.25, -0.5) {\scriptsize $l_n$};
		\node [style=BLACK] (14) at (1.5, 3) {};
		\node [style=BLACK] (15) at (1.5, 2.25) {};
		\node [style=BLACK] (16) at (1.5, 1.25) {};
		\node [style=BLACK] (20) at (1.5, -0.25) {};
		\node [style=none] (21) at (1.5, 0.5) {{\scriptsize $...$}};
		\node [style=none] (22) at (2.25, -0.25) {{\scriptsize $r_n$}};
		\node [style=none] (23) at (2.25, 3) {\scriptsize $r_1$};
		\node [style=none] (25) at (1.5, 0.25) {};
		\node [style=none] (26) at (2.25, 1.25) {\scriptsize $r_j$};
		\node [style=none] (27) at (-0.5, 2.5) {};
		\node [style=none] (29) at (-0.5, -0.25) {};
		\node [style=none] (36) at (0, 2.25) {\scriptsize $R$};
		\node [style=none] (37) at (-0.25, 2.75) {};
		\node [style=none] (38) at (0, 1.25) {\scriptsize $0$};
		\node [style=none] (42) at (-0.5, 0.5) {};
		\node [style=none] (43) at (-0.75, 0.25) {};
		\node [style=none] (44) at (-0.75, -0.5) {};
		\node [style=BLACK] (45) at (-1.5, -0.5) {};
		\node [style=none] (46) at (-1, -0.5) {};
		\node [style=none] (47) at (1, -0.25) {};
		\node [style=none] (48) at (1, 0) {};
		\node [style=none] (49) at (1.25, 0.25) {};
		\node [style=none] (53) at (0.75, 2.5) {};
		\node [style=none] (55) at (1, 2.75) {};
		\node [style=none] (56) at (1.25, 1.75) {};
		\node [style=none] (57) at (1, 1.5) {};
		\node [style=none] (62) at (1.5, 0.75) {};
		\node [style=none] (63) at (7.75, 3) {};
		\node [style=none] (64) at (7.75, 2) {};
		\node [style=none] (65) at (12.25, -0.25) {};
		\node [style=none] (66) at (12.25, 3.75) {};
		\node [style=BLACK] (67) at (9, 3.5) {};
		\node [style=GRAYN] (68) at (9, 1.25) {};
		\node [style=BLACK] (69) at (9, 2.5) {};
		\node [style=none] (70) at (8.25, 2.5) {\scriptsize $l_i$};
		\node [style=none] (71) at (8.25, 3.5) {\scriptsize $l_1$};
		\node [style=none] (72) at (8.25, 1.25) {\scriptsize $l_{\lceil \alpha n \rceil}$};
		\node [style=BLACK] (73) at (12, 3.5) {};
		\node [style=BLACK] (74) at (12, 2.75) {};
		\node [style=BLACK] (75) at (12, 1.75) {};
		\node [style=BLACK] (76) at (12, 0.25) {};
		\node [style=none] (77) at (12, 1) {{\scriptsize $...$}};
		\node [style=none] (78) at (12.75, 0.25) {{\scriptsize $r_n$}};
		\node [style=none] (79) at (12.75, 3.5) {\scriptsize $r_1$};
		\node [style=none] (81) at (12.75, 1.75) {\scriptsize $r_j$};
		\node [style=none] (82) at (10, 3) {};
		\node [style=none] (83) at (9.75, 1) {};
		\node [style=none] (84) at (10.5, 3) {\scriptsize $R$};
		\node [style=none] (85) at (10.25, 3.25) {};
		\node [style=none] (86) at (10.75, 2.25) {\scriptsize $0$};
		\node [style=none] (87) at (10, 2) {};
		\node [style=none] (88) at (10, 1.5) {};
		\node [style=none] (89) at (9.5, 1.25) {};
		\node [style=BLACK] (90) at (9, 1.25) {};
		\node [style=none] (91) at (9, 1.5) {};
		\node [style=none] (92) at (11.5, 0.25) {};
		\node [style=none] (93) at (11.5, 0.5) {};
		\node [style=none] (94) at (11.75, 0.75) {};
		\node [style=none] (95) at (11.25, 3) {};
		\node [style=none] (96) at (11.5, 3.25) {};
		\node [style=none] (97) at (11.75, 2.25) {};
		\node [style=none] (98) at (11.5, 2) {};
		\node [style=GRAYN] (100) at (9, 0.75) {};
		\node [style=GRAYN] (101) at (9, -0.75) {};
		\node [style=none] (102) at (8.25, -0.75) {\scriptsize $l_n$};
		\node [style=none] (103) at (9, 0) {{\scriptsize $...$}};
		\node [style=none] (104) at (9.5, -0.25) {};
		\node [style=none] (105) at (9.5, -0.5) {};
		\node [style=none] (106) at (9.5, -0.75) {};
		\node [style=none] (107) at (9.5, 0.5) {};
		\node [style=none] (108) at (9.5, 0.75) {};
		\node [style=none] (109) at (10, 0.25) {\scriptsize $0$};
		\node [style=none] (111) at (-0.5, 4.5) {{\scriptsize \textbf{a)} \mstLong}};
		\node [style=none] (112) at (11.5, 4.5) {{\scriptsize \textbf{b)} \mwpmLong}};
		\node [style=none] (113) at (5, 4.5) {};
		\node [style=none] (114) at (5, -1.25) {};
	\end{pgfonlayer}
	\begin{pgfonlayer}{edgelayer}
		\draw [style=striped-ultragray] (58.center)
			 to [in=-180, out=0, looseness=1.50] (61.center)
			 to (60.center)
			 to [in=0, out=180, looseness=1.50] (59.center);
		\draw [style=EDGE, in=180, out=30, looseness=0.75] (7) to (15);
		\draw [style=NiceRedThick, in=-150, out=15] (7) to (16);
		\draw [style=NiceRedThick] (0) to (14);
		\draw [style=NiceRedThick] (14) to (15);
		\draw [style=NiceRedThick] (15) to (16);
		\draw [style=NiceRedThick] (25.center) to (20);
		\draw [style=EDGE] (27.center) to (0);
		\draw [style=EDGE] (6) to (29.center);
		\draw [style=NiceRedThick, in=-120, out=30] (6) to (16);
		\draw [style=EDGE] (37.center) to (0);
		\draw [style=EDGE] (43.center) to (7);
		\draw [style=EDGE] (7) to (42.center);
		\draw [style=EDGE] (45) to (44.center);
		\draw [style=EDGE] (20) to (48.center);
		\draw [style=EDGE] (47.center) to (20);
		\draw [style=EDGE] (49.center) to (20);
		\draw [style=EDGE] (53.center) to (15);
		\draw [style=EDGE] (55.center) to (15);
		\draw [style=EDGE] (56.center) to (16);
		\draw [style=EDGE] (57.center) to (16);
		\draw [style=SLIGHTLYDASHED] (59.center) to (58.center);
		\draw [style=NiceRedThick] (62.center) to (16);
		\draw [style=striped-ultragray] (63.center)
			 to [in=-180, out=0, looseness=1.50] (66.center)
			 to (65.center)
			 to [in=0, out=180, looseness=1.50] (64.center);
		\draw [style=EDGE, in=180, out=15, looseness=0.75] (69) to (74);
		\draw [style=NiceRedThick, in=-150, out=15] (69) to (75);
		\draw [style=NiceRedThick] (67) to (73);
		\draw [style=EDGE] (82.center) to (67);
		\draw [style=EDGE] (68) to (83.center);
		\draw [style=EDGE] (85.center) to (67);
		\draw [style=EDGE] (88.center) to (69);
		\draw [style=EDGE] (69) to (87.center);
		\draw [style=EDGE] (90) to (89.center);
		\draw [style=EDGE] (76) to (93.center);
		\draw [style=EDGE] (92.center) to (76);
		\draw [style=EDGE] (94.center) to (76);
		\draw [style=EDGE] (95.center) to (74);
		\draw [style=EDGE] (96.center) to (74);
		\draw [style=EDGE] (97.center) to (75);
		\draw [style=EDGE] (98.center) to (75);
		\draw [style=SLIGHTLYDASHED] (64.center) to (63.center);
		\draw [style=NiceRedThick, bend left=15, looseness=1.25] (74) to (90);
		\draw [style=NiceRedThick, bend right=60, looseness=0.75] (100) to (76);
		\draw [style=EDGE] (107.center) to (100);
		\draw [style=EDGE] (108.center) to (100);
		\draw [style=EDGE] (104.center) to (101);
		\draw [style=EDGE] (101) to (105.center);
		\draw [style=EDGE] (106.center) to (101);
		\draw [style=EDGEDASHED] (113.center) to (114.center);
	\end{pgfonlayer}
\end{tikzpicture}

%% file: content/worst-case-bounds/minimum-weight-perfect-matching.tex
\section{Minimum-Weight Perfect Matchings}\label{sec:mwpm}
As any perfect matching has size $n/2$, the upper bound in \cref{th:ub-mst} immediately carries over to \mwpmLong as well.
We now show that the resulting additive error of $\Theta\big((n/\epsilon) \cdot \ln \big(\frac{m}{n} \big)\big)$  is tight for this problem as well (for sufficiently small $\delta$).
The encoding will be slightly different for MST because all coordinates with the same value would be encoded by the same vertex, and a matching can only recover one of those.
To fix this, we prove a standard balls and bins result stating that if $x$ has length $n' = \alpha n$ for some small $0 <\alpha\leq 1$, only a $\alpha^2$ fraction of a random $\X \sim \operatorname{Uni}([n]^{\lceil\alpha n\rceil})$ has collisions with high probability. 
%The proof to the following statement uses a standard concentration inequality and is deferred to \cref{app:constant-fraction}.
\begin{restatable}[Collision-free]{lemma}{constantfraction}\label[lemma]{lem:constant-fraction}\label{lem:collision-free}
Let $(X_1,\dots,X_d) \sim  \operatorname{Uni}\big([n]^d\big)$ with $d=\lceil \alpha n\rceil$ for a constant $0<\alpha\le1$. 
Then, with probability at least $1-\exp(-\Omega(\alpha^3 n))$, the number of indices $i$ for which there exists some $j\neq i$ with $X_i=X_j$ is $\mathcal{O}(\alpha^2 n)$.
\end{restatable}

\begin{proof}
Define $\tilde{X}_i=\mathbb{1}\left[\exists\, j\neq i: X_j=X_i\right]$ and let \(Y=\sum_{i=1}^d\tilde{X}_i\). 
By a union bound, \(\mathbb{E}[\tilde{X}_i]\le\sum_{j\neq i}\Pr{X_j=X_i}=(d-1)/n\), so that $\mathbb{E}[Y]\le d(d-1)/n = \alpha^2 n-\alpha\le\alpha^2 n$.
Since changing a single $X_i$ affects at most two of the \(\tilde{X}_i\)'s, \(Y\) is 2-Lipschitz. By McDiarmid's inequality \citep{McDiarmid_1989}, for any \(t>0\), $\Pr{Y\ge\mathbb{E}[Y]+t}\le\exp(-2t^2/(4d))$.
Setting $t=\mathbb{E}[Y]$ gives $\Pr{Y\ge2\mathbb{E}[Y]}\le\exp(-\mathbb{E}[Y]^2/(2d))=\exp(-\Omega(\alpha^3 n))$.
Thus, with probability at least \(1-\exp(-\Omega(\alpha^3 n))\), we have $Y\le2\mathbb{E}[Y]=\cO(\alpha^2 n)$.
\end{proof}

%\drafty{
%In other words, inside a small constant fraction $\alpha$ of the dataset $\X$, only a constant fraction $\alpha^2$ has collisions with high probability.\lukasnote{Maybe remove as duplicate.}
%Therefore, we can get a collision-free subset of the dataset of size roughly $\alpha n - \alpha^2 n \in \Theta(n)$ which is enough for our previous technique to work in this setting.
%}

%Each vertex $v_i$ on the right side represents the values in $[n]$ and will be mapped to the coordinates represented by the left part.
%Assuming that no two coordinates have the same value, it is clear that a perfect matching allows to reconstruct $\X$, but if they have, the previous construction does not work any more.
%Therefore, we cut the dataset $\X' = (X_1, \cdots, X_{n'})$ for some $n' =  \alpha n \in \Theta(n)$ for some $0 < \alpha < 1$ and only encode those coordinates that don't collide with any other coordinate which is roughly $\alpha n  -\alpha^2 n$, still a constant fraction of the dataset with high probability.
%being perfect sets as in front of a challenge, namely that you can not encode the same 

%Given a weighted graph $G$, the problem asks for a subset $M \subseteq E$ where each vertex $v\in E$ is adjacent to exactly one edge in $M$ while minimizing $w(M) = \sum_{m \in M}w(m)$.

\begin{restatable}[Lower bound MWPM]{theorem}{lbmatching}\label{th:lb-matching}
Let $\epsilon > 0$, $\delta \leq (n/m)^{\Omega(1)}$, and $m>2n$.
Then there exists a graph topology $F=(V,E)$ and a distribution of weights $\cD_\omega$ on weights of $E$ such that for any $(\epsilon, \delta)$-DP protocol  $\cB$ that outputs an approximate \emph{minimum-weight perfect matching} $M$ under the $\ell_1$ neighboring relationship, if the weights $\W \sim \cD_\omega$, then the expected additive error satisfies
\[ 
 w(M) - w(M^*) \geq \Omega((n/\epsilon) \cdot \min(\ln(1/\delta), \ln(m/n)))\,,
\] where $M^*$ is an optimal minimum-weight perfect matching.
\end{restatable}

\begin{proof}[Proof sketch.]
Using a similar encoding and observation as before, we can show that for some suitable constant $0< \alpha < 1$, for some an encoded random vector $\vec{x}\in \mathbb{R}^{\lceil \alpha n\rceil}$, at least $\alpha n (0.99 - \alpha)$ coordinates could be reconstructed in expectation.
This stands in contradiction with an allowed additive error of $\leq \alpha n^{-\Omega(1)}$ imposed by \cref{lem:reidentification} for $\delta \leq n^{-\Omega(1)}$. 
Finally, the same generalization technique to $m$ edges works here again.
Find the full proof in \cref{apx:proof-lb-mwpm}.
\end{proof}

%% file: content/hc.tex
\section{Hierarchical Clustering under Dasgupta's Cost Function}\label{sec:rc-hc}

In this section, we design the first reconstruction-attack based lower bound for hierarchical clustering. Observe that our approach for MST, where the dataset is encoded into every edge of the topology, does not work here. This is because, by a simple concentration argument, every cut in the graph will have weight approximately half the total size of the cut, preventing any cheap hierarchical clustering from existing. Instead, we use a sparse dataset encoding, which requires new technical ideas.

Before describing the attack, we introduce new graph notation. 
For a topology $F = (V,E)$, let $d_{\max}(F)$ denote the maximum degree  vertex in $F$.
Also, let $\phi(F)$ denote the minimum-size balanced cut in $F$; i.e. $\phi(F) = \min_{A,B \text{ balanced cut}} |E(A,B)|$ where $|E(A,B)|$ denotes the number of edges crossing the cut. For any $n$ and $d \geq 3$, there exist \emph{spectral expander graphs} where every node has degree $d$ and $\phi(F) \geq \Omega(nd)$~\citep{Vadhan2012}. These are the graphs where our bounds will be tightest.
%$\phi(F)$ is closely related to the spectra of $F$. Specifically, we say that a $d$-regular graph $F$ has spectral expansion $\lambda$ if $1-\lambda_2\big(\frac{1}{d}\textbf{A}_F\big) \geq \lambda$, where $\lambda_2(\cdot)$ denotes the second-largest eigenvalue of a matrix, and $\textbf{A}_F$ is interpreted as the adjacency matrix of $F$. Then, by Cheeger's inequality, it holds that $\phi(F) \geq \Omega(\lambda nd)$. Furthermore, for any constant $d$, there exist spectral expander graphs satisfying $\lambda \geq \Omega(1)$~ \cite{Vadhan2012}, where our bounds will be tightest.%\lukasnote{Shouldn't it rather be $\geq \Omega(n)$.}

To ensure the existence of a low-cost hierarchical clustering, we will encode a random dataset into a random \emph{sparse} subset, where each edge is sampled with probability $\leq 1/d_{\max}(F)$.
After privately computing the hierarchical clustering, we use the \emph{induced balanced-cut} (\cref{lem:hc-to-bc}) to try to reconstruct edges carrying encoded data that cross this cut; these edges will form the set $I$.
The error guarantee is sufficient to establish that $\EX[\sum_{i \in I} x_i]$ is low. However, because the encoding is now sparse, it becomes harder to argue that $\Ex{\abs{I}}$ is sufficiently large; \revise{i.e. that there are sufficiently many sampled edges crossing the cut which carry a $0$. In principle, a low-weight hierarchical clustering could also avoid many sampled edges and drive $|I|$ low.}

To circumvent this, we make the observation that the clustering algorithm cannot distinguish sampled vs. non-sampled edges among the edges of weight $0$, and thus the number of sampled edges in the cut carrying $0$ is not much lower than its expectation by another concentration argument. 
This enables us to prove a lower bound of $\Omega\big(\frac{\phi(F)}{\epsilon d_{\max}(F)}n\big)$ that grows with the \emph{minimum} balanced cut size. 
Somewhat counterintuitively, the lower bound shrinks, rather than grows, with the maximum degree, which is a byproduct of our sparse encoding. 
Nonetheless, this bound applies to many topologies, and is tight up to a logarithmic factor for sparse expanders (see discussion in Section~\ref{sec:lb-proof}).

\subsection{Reconstruction Attack Outline}

Our reconstruction attack encodes the dataset into a random subsample of edges $E' \subseteq E$, where each edge $e \in E$ is selected i.i.d. with probability $p = \frac{0.5}{d_{\max}(F)}$.
The dataset is encoded by taking an arbitrary bijection $\ell : E' \rightarrow [k]$ (a label function) where $k = |E'|$, and setting $w(e) := R\cdot x_{\ell(e)}$ for all $e \in E'$, and $w(e) := 0$ for all $e \in E \setminus E'$. 
This sparse encoding ensures that there is a low-cost hierarchical clustering.
Our reconstruction attack will not use the entire returned hierarchical clustering $T$ to reconstruct data. 
Instead, it will use only the edges in the induced balanced cut of $T$, which is defined in the following lemma from~\cite{dasgupta2016cost}. 
Their proof is constructive and forms the balanced cut with a procedure similar to a heavy-light decomposition.
We refer to the balanced cut guaranteed by \cref{lem:hc-to-bc} as the \emph{induced} balanced cut of $T$.
\begin{lemma}[{\cite[Lemma 11]{dasgupta2016cost}}]\label[lemma]{lem:hc-to-bc}
    For any hierarchical clustering $T$, there exists a balanced cut $A,B$ with $|A|, |B| \in (\frac{n}{3}, \frac{2n}{3})$ such that $\frac{w(A,B)}{\abs{A} \cdot \abs{B}} \leq \frac{27}{4n^3} \cdot \hccost{w}(T)$. 
    Moreover, there exists a procedure $\mathsf{HCtoBC}(T)$ which outputs the cut in linear time.
\end{lemma}
%If we were to encode a uniformly random dataset $\x \sim \{0,1\}^m$ into $w$, then with high probability, every balanced cut $A,B$ would have weight $w(A,B) \geq \frac{\phi(F)}{4}$ resulting in ${\hccost{w}(T^*) \geq \frac{\phi(F)}{4}\cdot n}$. This already prevents a satisfying additive lower bound for spectral expander graphs, where even if every edge was assigned weight $1$, the cost of hierarchical clustering would scale with $\cO(\phi(F)n \log(n))$.

The attack proceeds by running a hierarchical clustering algorithm $\cA_\hc$ on the graph $(V,E,w)$, obtaining a low-cost hierarchical clustering $T$, and computing the balanced cut $A,B$ of $T$ (which, by Lemma~\ref{lem:hc-to-bc}, is also guaranteed to have low cost). 
Then, it uses the reconstruction set $I = \{\ell(e) : e \in E'(A,B)\}$ (i.e. those data elements which are encoded to edges that cross $A,B$), and guesses $y_i = 0$ for every $i \in I$. 
This reconstruction procedure, $\ra^{\hc}$, appears in Algorithm~\ref{alg:ra-hc}.

\begin{algorithm}[t]
\caption{Reconstruction Attack $\ra^{\hc}$ for Hierarchical Clustering}\label{alg:ra-hc}
    \begin{algorithmic}[1]
    \REQUIRE{Public graph topology $F = (V, E)$, dataset $\x \in \{0,1\}^m$, parameter $R$, HC algorithm $\cA_{\hc}$}
        \STATE{Initialize weight function $w:E \rightarrow \{0, R\}$}.
        \STATE{Subsample $E'\subseteq E$ by selecting each edge from $E$ with $p = \probp$. 
        Let $k = \abs{E'}$.}
        \STATE{Let $\ell : E' \rightarrow [k]$ be any bijection from $E'$ to $[k]$ labels.}
        \STATE Set $w(e) := R\cdot x_{\ell(e)}$ for all $e \in E'$ and $w(e) := 0$ for all $e \in E \setminus E'$.\COMMENT{Encode $\x$ into the weights}
        \STATE{Compute hierarchical clustering $T:= \cA_{\hc}\big((F, w)\big)$.}
        \STATE{Compute induced balanced cut $(A,B) := \mathsf{HCToBC}(T)$.}
        \COMMENT{See \cref{lem:hc-to-bc}}
        \STATE{Initialize reconstruction set $I = \{\ell(e) : e \in E'(A,B)\} \subseteq [n]$.}
        \STATE{\textbf{for each} $i \in I$, guess $y_{i} = 0$.}
        \RETURN{Reconstruction set $I$, reconstructed vector $\vec{y} = \{y_i : i \in I\}$.}
    \end{algorithmic}
\end{algorithm}

An important observation for our analysis is the following: The edges $e \in E'$ such that $x_{\ell(e)} = 0$ appear the same as the edges $e \in E \setminus E'$ in the weight function (both appear as $0$), and thus the returned tree $T$ gives us no further information to which of these two sets an edge of weight $0$ belongs. Formally, define $E_0' = \{ e \in E' : x_{\ell(e)} = 0\}$ and $\{E_1' = \{e \in E' : x_{\ell(e)} = 1\}$. The observation can be stated as:
\begin{observation}\label[observation]{obs:independence}
    In Algorithm~\ref{alg:ra-hc}, the tree $T$, and subsequent post-processing of it, are conditionally independent of $E_0'$ given $E_1'$.
\end{observation}
We will use this observation to show that there are many edges in $E_0'(A,B)$, just due to the randomness of $E'$, and thus $I$ will be sufficiently large.

\subsection{Reconstruction Attack Lower Bound}\label{sec:lb-proof}
We will now prove the following reconstruction attack.
\begin{restatable}[Lower Bounds HC]{theorem}{main}\label{thm:hc-reconstruction}
    Let $F = (V,E)$ be a public graph topology such that $d_{\max}(F) \geq 8$, and $\phi(F) \geq \max\{8n, 10000  d_{\max}(F) \log n\}$.
    Suppose $\hc$ is a hierarchical clustering algorithm with additive error $\leq R \cdot \frac{n\phi(F)}{400 d_{\max}(F)}$ for some $R > 0$. Then, $\suc(\ra^{\hc}) \geq 0.4$. Furthermore, the expected size of the reconstructed set is at least $\tfrac{\phi(F)}{8 d_{\max}(F)}$.
%\end{theorem}
\end{restatable}
This yields the following immediate corollary whose proof is identical to Corollary~\ref{coro:mst-lb}:
\begin{corollary}\label[corollary]{coro:hc-lb}
    Let $F = (V,E)$ be a public graph topology satisfying the conditions of \Cref{thm:hc-reconstruction}. Then, for any $\epsilon \leq 1, \delta \leq \frac{0.01\epsilon}{m}$, there is no $(\epsilon, \delta)$-DP algorithm for hierarchical clustering with error less than $\frac{n \phi(F)}{400\epsilon d_{\max{(F)}}}$.
\end{corollary}

Corollary~\ref{coro:hc-lb} is tightest when $F$ is a $d$-regular spectral expander graph where $\phi(F) \geq \Omega(nd)$. 
In this case, the lower bound of $\Omega\big(\frac{n^2}{\epsilon}\big)$ is the same as that in~\citet{deng2025pricedifferentialprivacyhierarchical} (for the complete graph only) and is tight with the upper bound in~\citet{pmlr-v202-imola23a} up to $\log(n)$ factors.
Observe that our lower bound argument can be applied to any \emph{subgraph} of $F$ by just ignoring edges and vertices, and thus it is possible to remove a small number of high-degree nodes and derive a stronger lower bound on a subgraph.

We now prove Theorem~\ref{thm:hc-reconstruction} and therefore need some supporting lemmas. 
First, we will show the optimal cost of clustering the encoded graph is just $\cO(n \cdot \log(n))$ with high probability. This comes from the fact that the components in $E'$ w.h.p. will have size $\cO(\log(n))$.

%% TODO for next version. Check these constants carefully.
%\begin{lemma}[Existence of a good Dasgupta cluster]
\begin{restatable}[Existence of a good Dasgupta cluster]{lemma}{goodcluster}
\label[lemma]{lem:good-cluster}
Suppose $F = (V,E)$ is a graph and let $\alpha \in (0,1)$ be a constant. 
Now, assign weights $w(e) := 1$ with probability $\frac{1-\alpha}{d_{\max}(F)}$ and $w(e) := 0$ otherwise for each edge $e \in E$ independently.
Then with probability at least $1-2n^{-3}$, there exists a clustering tree $T$, such that $\hccost{w}(T) \leq \frac{2-\alpha}{\alpha^2} \cdot 4 n \log(n)$.
\end{restatable}
\ifdefined\PODS The proof appears in the full version of the paper~\cite{todo:arxiv-version}. \else
The proof appears in Appendix~\ref{app:good-cluster}. \fi The number of mistakes made in the set $I$ is given by $\sum_{i \in I} x_i = |E_1'(A,B)|$. We use the previous bound on the optimal clustering cost to convert the additive error of $\hc$ into an absolute bound on $|E_1'(A,B)|$.

\begin{lemma}\label[lemma]{lem:hc-cut}
    Suppose Algorithm~\ref{alg:ra-hc} is instantiated with an algorithm $\cA_{\hc}$ with expected additive error 
    at most $C_F Rn$ for a constant $C_F$ (that may depend on $n, F$).
    Then, the balanced cut $(A,B)$ returned by $\mathsf{HCToBC}(T)$ has expected size $\EX[\abs{E_1'(A,B)}] \leq 2 C_F + 50 \log(n)$. 
\end{lemma}
\begin{proof}
    By the additive error guarantee of $\hc$, we have for fixed weights $w$,
    $\Ex{\hccost{w}(T)|w} \leq \hccost{w}(T^*) + C_F R n$.
    Recall that $R > 0$ is the weight, that we put on the subsampled edges that encode an one of the dataset.
    Note that in \cref{alg:ra-hc}, we use $\alpha = 0.5$ and then $\frac{2-\alpha}{\alpha^2} = 6$.
    By Lemma~\ref{lem:good-cluster}, with probability at least $1-2n^{-3}$, we have that $w$ has an optimal clustering cost of $(6 \cdot 4) \cdot  R n \log n$. 
    In the worst case, the cost of clustering a complete graph is $n^3 R$.
    Thus, we can use conditional expectation to show $\Ex{\hccost{w}(T)} \leq 24 n \log(n) + 2 n^{-3} n^3R + C_F R n \leq 25 R n \log(n) + C_F R n$. 
    Furthermore, by \cref{lem:hc-to-bc}, the balanced cut $A,B$ satisfies
    \begin{align*}
        w(A,B) \leq \frac{27}{4n^3} \lvert A \rvert \lvert B \rvert \hccost{w}(T) \leq \frac{27}{16 n} \hccost{w}(T)\,,
    \end{align*}
    because $\abs{A}\abs{B}$ maximizes in the interval $(\frac{n}{3}, \frac{2n}{3})$ exactly if both have size exactly  $\frac{n}{2}$.
    The last step is to observe that $w(A,B) = R \cdot  \abs{E_1'(A,B)}$ due to the way the data is embedded.
    Hence, we get the following in expectation,
    \begin{align*}
        R\cdot \Ex{\abs{E_1'(A, B)}} &= \Ex{w(A,B)} \leq \frac{27}{16 n} \Ex{\hccost{w}(T)} \leq \frac{27}{16 n} \left( C_F R n + 25 Rn\log n \right)\\
        & \leq R (2C_F + 50 \log n).
    \end{align*}
    Dividing by $R$ yields the result.
\end{proof}
Next, we need to show that $|I|$ is sufficiently large, as this will show that $I$ is a successful reconstruction set. As $\abs{I} = \lvert E_0'(A,B)\rvert + \lvert E_1'(A,B)\rvert$, and we have already upper bounded $|E_1'(A,B)|$, we will lower bound $\lvert E_0'(A,B)\rvert$. Precisely, we lower bound $\Ex{\lvert E_0'(A,B)\rvert \middle\vert E_1'}$ since conditioning on $E_1'$ allows us to consider $A,B$ as a constant per Observation~\ref{obs:independence}. Under this conditioning, the edges in $E_0'$ are distributed multinomially, allowing us to do a straightforward Chernoff bound.
Find the proof for the \cref{lem:cut-lb,lem:bc-mult-shrink} in \Cref{app:cut-proofs}.

\begin{restatable}{lemma}{cutlb}\label[lemma]{lem:cut-lb}
    Suppose $\x$ and $E'$ are sampled as they are in Algorithm~\ref{alg:ra-hc}. Then,
    for any balanced cut $A_0,B_0$, we have that $\Ex{\lvert E_0'(A_0,B_0)\rvert \middle\vert E_1'} \geq \frac{p}{2}\lvert \tilde E(A_0,B_0)\rvert$ where $\tilde E = E \setminus E_1'$.
\end{restatable}
The prior bound of ${|\tilde E(A_0,B_0)|}$ is not completely independent of $E_1'$; the following result shows it is almost always bigger than $\frac{2}{3}|E(A_0,B_0)|$, giving us a simple bound to work with.

\begin{lemma}\label[lemma]{lem:bc-mult-shrink}
    Suppose $d_{\max} \geq 8$ and $\phi(F) \geq 8 n$.
    With probability $1-2^{-0.2n}$, for all balanced cuts $A_0,B_0$, we have $\abs{E_1'(A_0,B_0)} \leq \frac{1}{3} \abs{E(A_0,B_0)}$.
\end{lemma}

We are now ready to put it all together and prove our main reconstruction lower bound for HC.
\main*
\begin{proof}
Observe that, based on the construction of the set $I$ and $y$, we have $\sum_{i \in I}\mathbb{1}[x_i \neq y_i] = E_1'(A,B)$ and $|I| = E_1'(A,B) + E_0'(A,B)$. We will prove that $9 \EX[E_1'(A,B)] \leq \EX[E_0'(A,B)]$, as this implies $\Ex{\sum_{i \in I}\mathbb{1}[x_i \neq y_i]} \leq 0.1 \Ex{\abs{I}}$, meaning the attack is $0.4$ successful.
Because $\hc$ has additive error $\frac{\phi(F)R}{400 d_{\max}(F)}n$, by \cref{lem:hc-cut}, we have $\EX[|E_1'(A,B)|] \leq \frac{\phi(F)}{200 d_{\max}(F)} + 50 \log(n)$.
By assumption, we know that $50 \log(n) \leq \frac{\phi(F)}{200 d_{\max}(F)}$, and so $\EX[|E_1'(A,B)|] \leq \frac{\phi(F)}{100 d_{\max}(F)}$.
By the law of total expectation, we have $\EX[\abs{E_0'(A,B)}] = \EX[\EX[\abs{E_0'(A,B)} \vert E_1']]$. 
We can bound the inner quantity,
    \begin{align*}
        \EX[\abs{E_0'(A,B)} \vert E_1'] &= \EX_{A,B}\left[\EX\left[\lvert E_0'(A,B)\rvert \middle \vert \,A,B,E_1' \right]\right] & \; \\
        &\geq \min_{A_0,B_0 \text{ balanced cut}}\EX\left[\lvert E_0'(A_0,B_0)\rvert \middle \vert A=A_0,B=B_0, E_1'\right] & \; \\
        &= \min_{A_0,B_0 \text{ balanced cut}}\EX\left[\lvert E_0'(A_0,B_0)\rvert \middle \vert E_1'\right] & \text{(by Observation~\ref{obs:independence})} \\
        &\geq \frac{p}{2}\min_{A_0,B_0 \text{ balanced cut}} \lvert (E \setminus E_1')(A_0,B_0)\rvert. & \text{(by Lemma~\ref{lem:cut-lb})}.
        %&= \frac{p}{2}\left(\phi(G) - \frac{mp}{2}\right).
    \end{align*}
    Now, we will expand the whole expectation. In the following, let $\mathcal{E}$ be the event in Lemma~\ref{lem:bc-mult-shrink}; i.e. that for all balanced cuts $A_0,B_0$, we have $|E_1'(A_0,B_0)| \leq \frac{1}{3} |E(A_0,B_0)|$.
    \begin{align*}
        \EX[\abs{E_0'(A,B)}] &\geq \EX\left[\frac{p}{2}\min_{A_0,B_0 \text{ balanced cut}} \lvert (E \setminus E_1')(A_0,B_0)\rvert\right] \\
        &\geq \Pr{\mathcal{E}}\EX\left[\frac{p}{2}\min_{A_0,B_0 \text{ balanced cut}} \lvert (E \setminus E_1')(A_0,B_0)\rvert\middle \vert \,\mathcal{E}\right] \\
        &\geq (1-2^{-0.2n})\frac{p}{2}\EX\left[\min_{A_0,B_0 \text{ balanced cut}} \frac{2}{3}\lvert E(A_0,B_0)\rvert\middle \vert \,\mathcal{E}\right] & \text{(by Lemma~\ref{lem:bc-mult-shrink})}\\
        &> \frac{p}{4} \phi(G) = \dfrac{\phi(F)}{8d_{\max}}.
    \end{align*}
   Putting it all together, recalling that $\EX[|E_1'(A,B)|] \leq \frac{\phi(F)}{100\, d_{\max}(F)}$ and $\EX[|E_0'(A,B)|] > \frac{\phi(F)}{8\, d_{\max}(F)}$, we obtain $9 \EX[\abs{E_1'(A,B)}] \leq \frac{9\,\phi(F)}{100\, d_{\max}(F)} < \frac{\phi(F)}{8\, d_{\max}(F)} < \EX[\abs{E_0'(A,B)}]$. \qedhere
\end{proof}

%% file: content/algorithms.tex
\section{Private Algorithms with Tight Error}
We now detail algorithms with superior error on MST and MWPM that appear in ongoing work~\citep{aamand-personal-communication} demonstrating that the lower bounds of Sections~\ref{sec:mst-wc} and~\ref{sec:mwpm} are tight.
The algorithm is a simple input privatization, adding Laplace noise to each of the edge weights independently. 
The novelty lies in the analysis: instead of bounding the maximum noise added to individual edge weights as in prior work~\citep{Sealfon_2016,hladik_2024,pagh_2025}, it  uses a concentration bound on the total noise of any single spanning tree and then applies a union bound.
We can phrase it even more general: 
any graph problem that outputs $k$ edges can achieve additive error $\cO\big(\frac{k}{\epsilon} \cdot \ln\PAREN{\frac{m}{k}}\big)$, because one can bound $\binom{m}{k} \leq (\frac{em}{k})^k$.
In the cases of MSTs and MWPMs where $k= \Theta(n)$, this gives an upper bound with a logarithmic dependence of $\ln\big(\frac{m}{n}\big)$.
\ifdefined\PODS A proof of this theorem appears in the full version~\cite{}. \else The proof appears in \Cref{app:ub-mst}. \fi

\begin{restatable}[{\cite{aamand-personal-communication}}]{theorem}{ubmst}\label{th:ub-mst}
    Let $G = (V, E, \vec{w})$ be a graph and let $\cS(G)$ be a family of solutions such that every $S \in \cS(G)$ satisfies $S \subseteq E$ and $|S| = k$.
    Consider the problem  of minimizing  $w(S)$ among all $S \in \cS(G)$.
    Then there exists an $\epsilon$-DP algorithm $\cA$ that outputs a feasible solution $S \in \cS(G)$, such that with probability at least $1-n^{-\Omega(1)}$,
    \[
        w(S) - w(S^*) \leq \cO\PAREN{\frac{k}{\epsilon} \cdot \ln\PAREN{\frac{m}{k}}},
    \]
    where $S^* = \arg \min\limits_{\tilde S \in \cS(G)} w(\tilde S)$ is any optimal solution.
%    There exists an $(\epsilon, \delta)$-DP algorithm $\cA$ that, on input $G = (V, E)$ with $m/n \geq e$, and weights $w : E \to \R_{\geq 0}$, outputs a subset of edges $S \subseteq E$ of exactly $0 < k < m$ edges, such that, with probability at least $1- n^{-\Omega(1)}$,
\end{restatable}

%% file: content/conclusion.tex
\section{Conclusion and Future Work}
%We have shown lower bounds for MST, MWPM, and HC. 
We first showed that for MST and MWPM there exist worst-case topologies where the error is tight with the upper bounds. 
We also proved lower bounds for MST and HC on large classes of topologies which are close to the upper bounds, off by just a $\log(n)$ factor. 

For MWPM, it remains open whether lower bounds can be derived for a large class of topologies, and what specific property that class should satisfy. 
For HC, closing the gap between the best lower bound of $\Omega\big(\frac{n^2}{\epsilon}\big)$ and upper bound of $\cO\big(\frac{n^2 \log(n)}{\epsilon}\big)$ is an interesting question, as well as deriving a polynomial-time algorithm that achieves this additive error (as opposed to using the exponential mechanism).
Finally, we believe that both our lower bound techniques for obtaining tight bounds and bounds on general topologies are likely to carry over to other graph problems as well.

%% file: content/appendix-full.tex
\section{Additional Details}\label{app:details}

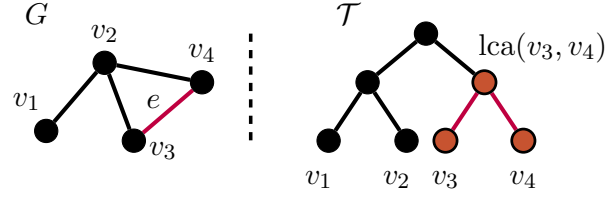
\begin{figure}[t]
    \centering
     \resizebox{0.5\columnwidth}{!}{
          \input{fig/clustering-example.tikz}}
    \caption{\textit{An example of a hierarchical clustering $T$ with $\hccost{w}(T) = 12$ for a graph $G = (V, E)$ assuming unit weights.
    The edge $e$ has cost $2$ as the lowest common ancestor $\operatorname{lca}(v_3, v_4)$ has exactly two leaves in its subtree of $T$.
    Furthermore, cutting the tree at $\operatorname{lca}(v_3, v_4)$ would give a balanced cut.
    }}
\label{fig:dsexamples}
\end{figure}

In Figure~\ref{fig:dsexamples}, we illustrate the Dasgupta cost function, $\hccost{w}(T)$, for a given hierarchical clustering tree $T$.

\section{Omitted Proofs}

\subsection{Proof of Lemma~\ref{lem:reidentification}}

\reident*
\begin{proof}
Let $n \in \mathbb{N}$.
We index the elements in $\cX$ by $\{1, \cdots, n\}$.
For some arbitrarily chosen $i\in [d]$, assume a random vector $\X = \Xleqi  \oplus X_i \oplus \Xgeqi \sim \operatorname{Uni}\PAREN{[n]^d}$ where $\oplus$ denotes vector concatenation, and we use the subscript $\Xleqi$ and $\Xgeqi$ to split the vector at index $i$. 
Now we can bound the probability that the $i$'th coordinate of the output of $\cB(\X)$ leaks its corresponding input $x_i$:

\begin{align}
\Pr{\cB(\X)_i = x_i} &= 
\dfrac{1}{{|\cX|}^{d-1}}\sum_{\Xleqi\in [n]^{i-1}} 
\sum_{\Xgeqi \in [n]^{d-i}}\left(\Pru{X_i \sim \operatorname{Uni}\PAREN{[n]}}{\cB(\Xleqi \oplus x_i \oplus \Xgeqi)_i = x_i}\right)\nonumber\\
& \leq {|\cX|}^{1-d}\sum_{\Xleqi\in [n]^{i-1}} \label{line:enumerate}
\sum_{\Xgeqi \in [n]^{d-i}}\left(e^\epsilon \Pru{x_i \sim \operatorname{Uni}\PAREN{[n]}}{\cB(\Xleqi \oplus (1) \oplus \Xgeqi)_i= x_i} + \delta\right)\\
& \leq {|\cX|}^{1-d}\sum_{\Xleqi\in [n]^{i-1}}  \sum_{\Xgeqi \in [n]^{d-i}}\left(e^\epsilon \dfrac{1}{n} + \delta\right)\label{line:fix}\\
&\leq \frac{e^\epsilon}{|\cX|}+ \delta\,.\nonumber
\end{align}
In line~\ref{line:enumerate}, we use the fact that each $X_i$ is i.i.d. and in step \ref{line:fix}, we use the privacy guarantees of the mechanism $B$ and flip to a neighboring dataset where we explicitly set $X_i = 1$ under the hamming adjacency relation.
\end{proof}

\subsection{Proof of Theorem~\ref{lem:reconst}}
First, we start with a supporting lemma.
\begin{lemma}\label{lem:approx-dp}
    Suppose $P,Q$ are distributions on a discrete set $\cX$ satisfying 
    \begin{align*}
    \Pru{X \sim P}{X \in S} \leq e^{\epsilon} \Pru{Y \sim Q}{Y \in S} + \delta \text{~for all~} S \subseteq \mathcal{X}\,.
    \end{align*}

    Then, there exists a function $\delta(x) : \cX \rightarrow [0, \infty)$ satisfying 
    \begin{enumerate}
        %\item $P(x) \leq e^\epsilon Q(x) + \delta(x)$\,, and 
        \item $\Pru{X \sim P}{X = x} \leq e^\epsilon \Pru{X \sim Q}{X = x} + \delta(x)$\,, and \label{condone}
        \item $\sum_{x \in \cX} \delta(x) \leq \delta$\,.\label{condtwo}
    \end{enumerate}
\end{lemma}
\begin{proof}
    We denote $P(x) = \Pru{X \sim P}{X = x}$ (and similar for $Q$).

    We take the function $\delta(x) = \max\{P(x) - e^\epsilon Q(x), 0\}$. 
    Condition~\ref{condone} is satisfied because $P(x) = e^\epsilon Q(x) + (P(x) - e^\epsilon Q(x)) \leq e^\epsilon Q(x) + \delta(x)$. 
    
    For condition~\ref{condtwo}, let $N \subseteq \mathcal{X}$ denote the set where $P(x) - e^\epsilon Q(x) \geq 0$. On $x \in N$, we have $\delta(x) = P(x) - e^\epsilon Q(x)$, and on $x \in \mathcal{X} \setminus N$, we have $\delta(x) = 0$.
    Thus,
    \begin{align*}
        \sum_{x \in \cX} \delta(x) 
        =  \sum_{x \in N} \delta(x) = \sum_{x \in N} P(x) - e^\epsilon Q(x) 
        = \Pru{X \sim P}{X \in N} - e^\epsilon \Pru{Y \sim Q}{Y \in N} \leq \delta\,.
    \end{align*}
\end{proof}

\reconst*
\begin{proof}\label{proof:subsetrec} It suffices to show that
    \begin{align*}
    \Ex{\sum_{i \in I} \mathbb{1}[x_i \neq y_i]} \geq \frac{1}{1+e^\epsilon}\Ex{\abs{I}} - \frac{1}{1+e^\epsilon}d \delta\,.
    \end{align*}
    We use conditional expectation to write
    \begin{align*}
        \Ex{\sum_{i \in I}\mathbb{1}[x_i \neq y_i] }
        = \Ex{\Ex{\sum_{i \in I}|y_i-x_i|\middle \vert I}} = \Ex{\sum_{i \in I}\Ex{\abs{y_i-x_i}\vert I}}\,.
    \end{align*}
    We will form a bound $\Ex{\abs{y_i - x_i} \vert I}$ for each $i$. 
    To do this, we expand it as follows by conditional probability,

    \begin{align*}
        \Ex{|y_i-x_i| \vert I } &= \Pr{x_i = 1, y_i = 0 \vert  I} + \Pr{x_i = 0, y_i = 1 \vert I} \\ 
        &= \frac{\Pr{x_i = 1, y_i = 0, I} + \Pr{x_i=0, y_i=1, I}}{\Pr{I}}.
    \end{align*}
    In the following, a sum over the variable $x$ implicitly ranges over $\{0,1\}^n$, and a sum over $y$ implicitly ranges over $\{0,1\}^I$. We can write 
    \begin{align*}
        \Pr{x_i=a, y_i=b, I} = \sum_{\substack{x:x_i=a \\ y:y_i=b}} 2^{-n} \Pr{\mathcal{A}(x) = (I,y)} = 2^{-n} \sum_{x:x_i=a} \Pr{\mathcal{A}(x) \in (I, \mathcal{Y}_{i,b})},
    \end{align*}
    where $\mathcal{Y}_{i,b} = \{y \in \{0,1\}^I : y_i = b\}$.
    For a vector $x \in \{0,1\}^n$, let $x^{(i \rightarrow 0)}$ and $x^{(i \rightarrow 1)}$ denote $x$ with its $i$th bit set to $0$ (resp. 1).
    Because each $\mathcal{A}(x^{(i\rightarrow 0)})$ and $\mathcal{A}(x^{(i \rightarrow 1)})$ satisfy $(\epsilon, \delta)$-DP, by Lemma~\ref{lem:approx-dp}, there exists a non-negative function $\delta_{x\uparrow i}(I,b)$ such that 
    \[
    \Pr{\mathcal{A}(x^{(i \rightarrow 0)}) \in (I,\mathcal{Y}_{i,b})} \leq e^\epsilon\Pr{\mathcal{A}(x^{(i \rightarrow 1)}) \in (I,\mathcal{Y}_{i,b})} + \delta_{x\uparrow i}(I,b),
    \]
    and $\sum_{I \subseteq [n]} \sum_{b \in \{0,1\}} \delta_{x\uparrow i}(I,b) \leq \delta$. 
    Similarly, there exists a non-negative function $\delta_{x \downarrow i}(I,b)$ such that 
    \[
    \Pr{\mathcal{A}(x^{(i \rightarrow 1)}) \in (I,\mathcal{Y}_{i,b})} \leq e^\epsilon \Pr{\mathcal{A}(x^{(i \rightarrow 0)}) \in (I,\mathcal{Y}_{i,b})} + \delta_{x \downarrow i}(I,b)
    \]
    and $\sum_{I \subseteq [n]} \sum_{b \in \{0,1\}} \delta_{x\downarrow i}(I,b) \leq \delta$.
    By summing the first inequality over $x \in \{0,1\}^n$ with $x_i = 0$ (and the second over $x$ with $x_i = 1$), we obtain
    \begin{align*}
        \Pr{x_i=0, y_i=0,I} &= 2^{-n}\sum_{x:x_i = 0} \Pr{\mathcal{A}(x) \in (I, \mathcal{Y}_{i,0})} \\
        &\leq 2^{-n}\sum_{x : x_i=1} \PAREN{e^\epsilon\Pr{\mathcal{A}(x) \in (I, \mathcal{Y}_{i,0})} + \delta_{x \uparrow i}(I,0)} \\
        &= e^\epsilon \Pr{x_i=1, y_i=0, I} + \delta_{i,1}^\uparrow(I,0),
    \end{align*}
    where $\delta_{i,a}^\uparrow(I,b)$ is defined to be $2^{-n} \sum_{x:x_i=a} \delta_{x \uparrow i}(I,b)$. Rearranging, we have
    \[
        \Pr{x_i=1, y_i=0, I} \geq e^{-\epsilon} \Pr{x_i=0, y_i=0, I} - e^{-\epsilon} \delta_{i,1}^\uparrow(I,0).
    \]
    Similarly, we can write
    \[
        \Pr{x_i = 0, y_i = 1, I} \geq e^{-\epsilon}\Pr{x_i=1, y_i=1, I} - e^{-\epsilon} \delta_{i,0}^\downarrow (I,1),
    \]
    where $\delta_{i,a}^\downarrow(I,b) = 2^{-n}\sum_{x:x_i=a}\delta_{x \downarrow i}(I,b)$. This means that
    \begin{align*}
        &\Pr{x_i = 1, y_i = 0, I} + \Pr{x_i=0, y_i=1, I} \\ 
        &= \frac{1}{1+e^\epsilon} (\Pr{x_i = 1, y_i = 0, I} + \Pr{x_i=0, y_i=1, I}) \\
        &\qquad+ \frac{e^\epsilon}{1+e^\epsilon} (\Pr{x_i = 1, y_i = 0, I} + \Pr{x_i=0, y_i=1, I}) \\
        &\geq \frac{1}{1+e^\epsilon} (\Pr{x_i = 1, y_i = 0, I} + \Pr{x_i=0, y_i=1, I}) \\
        &\qquad+ \frac{e^\epsilon}{1+e^\epsilon} (e^{-\epsilon} \Pr{x_i=0, y_i=0, I} - e^{-\epsilon} \delta_{i,1}^\uparrow(I,0) + e^{-\epsilon}\Pr{x_i=1, y_i=1, I} - e^{-\epsilon} \delta_{i,0}^\downarrow (I,1)) \\
        &= \frac{1}{1+e^\epsilon} (\Pr{x_i = 1, y_i = 0, I} + \Pr{x_i=0, y_i=1, I}) \\
        &\qquad+ \frac{1}{1+e^\epsilon} ( \Pr{x_i=0, y_i=0, I} - \delta_{i,1}^\uparrow(I,0) + \Pr{x_i=1, y_i=1, I} - \delta_{i,0}^\downarrow (I,1)) \\
        &= \frac{1}{1+e^\epsilon} \Pr{I} - \frac{1}{1+e^\epsilon} (\delta_{i,1}^\uparrow(I,0) + \delta_{i,0}^\downarrow (I,1)),
    \end{align*}
    where the last step follows from the law of total probability. Thus,
    \begin{align*}
        \Ex{|y_i-x_i|\middle \vert I } &= \frac{\Pr{x_i = 1, y_i = 0, I} + \Pr{x_i=0, y_i=1, I}}{\Pr{I}} \\
        &\geq \frac{1}{1+e^\epsilon} - \frac{\delta_{i,1}^\uparrow(I,0) + \delta_{i,0}^\downarrow (I,1)}{(1+e^\epsilon)\Pr{I}}.
    \end{align*}
    
    Finally, we bound the entire expectation as
    \begin{align}
        &\EX\left[\sum_{i\in I} \EX[|y_i-x_i| \vert I] \right] \geq \EX\left[\sum_{i \in I} \left(\frac{1}{1+e^\epsilon} - \frac{\delta_{i,1}^\uparrow(I,0) + \delta_{i,0}^\downarrow (I,1)}{(1+e^\epsilon)\Pr{I}}\right)\right] \nonumber \\
        &\qquad = \sum_{I \subseteq [n]}\Pr{I}\sum_{i \in I} \left(\frac{1}{1+e^\epsilon} - \frac{\delta_{i,1}^\uparrow(I,0) + \delta_{i,0}^\downarrow (I,1)}{(1+e^\epsilon)\Pr{I}}\right) \nonumber \\
        &\qquad= \frac{1}{1+e^\epsilon} \sum_{I \subseteq [n]} \Pr{I} \cdot |I| - \frac{1}{1+e^\epsilon} \sum_{I \subseteq [n]}\Pr{I} \sum_{i \in I} \frac{\delta_{i,1}^\uparrow(I,0) + \delta_{i,0}^\downarrow (I,1)}{\Pr{I}} \nonumber \\
        &\qquad = \frac{1}{1+e^\epsilon} \EX[|I|] - \frac{1}{1+e^\epsilon} \sum_{I \subseteq [n]}\sum_{i \in I} \delta_{i,1}^\uparrow(I,0) + \delta_{i,0}^\downarrow (I,1)\,. \label{eq:rcp3}
    \end{align}

    The second term simplifies to
    \begin{align*}
        &\sum_{I \subseteq [n]}\sum_{i \in I} \delta_{i,1}^\uparrow(I,0) + \delta_{i,0}^\downarrow (I,1) \\
        &\quad= \sum_{I \subseteq [n]}\sum_{i \in I} 2^{-n}\left(\sum_{x : x_i=1}\delta_{x \uparrow i}(I,0) + \sum_{x:x_i=0} \delta_{x\,\downarrow\, i}(I,1)\right) \\
        &\quad= \sum_{i=1}^n \sum_{I \subseteq [n]} \mathbb{1}[i \in I] \cdot 2^{-n}\left(\sum_{x : x_i=1}\delta_{x \uparrow i}(I,0) + \sum_{x:x_i=0} \delta_{x\,\downarrow\, i}(I,1)\right) \\
        &\quad\leq \sum_{i=1}^n \sum_{I \subseteq [n]} 2^{-n}\left(\sum_{x : x_i=1}\delta_{x \uparrow i}(I,0) + \sum_{x:x_i=0} \delta_{x\,\downarrow\, i}(I,1)\right) \\
        &\quad= \sum_{i=1}^n 2^{-n}\left(\sum_{x : x_i=1}\sum_{I \subseteq [n]}\delta_{x \uparrow i}(I,0) + \sum_{x:x_i=0} \sum_{I \subseteq [n]} \delta_{x\,\downarrow\, i}(I,1)\right) \\
        &\quad\leq \sum_{i=1}^n 2^{-n}\left(\sum_{x : x_i=1}\delta + \sum_{x:x_i=0} \delta \right) = \sum_{i=1}^n \delta = n\delta.
        \end{align*}
    Plugging back into \eqref{eq:rcp3}, we obtain
    \begin{align*}
        \EX\left[\sum_{i\in I} \EX[|y_i-x_i| \vert I] \right] &\geq \frac{\EX|I|}{1+e^\epsilon} - \frac{1}{1+e^\epsilon}n \delta,
    \end{align*}
    giving the answer.
\end{proof}

\begin{comment}
\paragraph{Probability Theory.} 
We recall the multinomial distribution and a standard conditioning identity that we use later.
\begin{definition}[Multinomial Distribution]
A random vector $(X_1,\dots,X_k)$ has multinomial distribution 
$\mnom(n,\mathbf{p})$
with parameters $n$ and probability vector $\vec{p}$ if the joint probability satisfies 
${\mathbb{P}(X_1=x_1,\dots,X_k=x_k) =  \frac{n!}{x_1!\cdots x_k!}\prod_{i=1}^k p_i^{x_i}}$ for $\sum_{i=1}^k x_i = n$ and
$0$ otherwise
\end{definition}
\begin{restatable}{lemma}{multinom}\label{lem:cond-mnom}
Let $(X_1,\cdots,X_k)\sim \mnom(n,\mathbf{p})$, where $\vec{p}$ is a probability vector. 
Then for any $j\in[k]$ and $m \in \{0, \cdots, n\}$,
\[
(X_1,\cdots,X_{j-1},X_{j+1},\cdots,X_k)\,\big|\, (X_j=m)
\sim
\mnom\left(
n-m,\,\frac{(p_1,\cdots, p_{j-1}, p_{j+1}, \cdots, p_k)}{1-p_j}
\right).
\]
\end{restatable}

\multinom*
\begin{proof}\label{proof:multinom}
By definition of conditional probability and the fact that $\sum_{i \neq j}x_i = n - m$, we have,
\begin{align*}
\Pr{ X_i = x_i, \text{~for all~} i \neq j \mid X_j = m} 
&= \dfrac{n!/m!\PAREN{\prod_{i\neq j}\frac{ p_i^{x_i}}{x_i!}}p_j^m}{{n \choose m}(1-p_j)^{n-m}p_j^{m}} 
= \dfrac{(n-m)!\PAREN{\prod_{i\neq j}\frac{p_i^{x_i}}{x_i!}}}{(1-p_j)^{n-m}} \\
&=\dfrac{(n-m)!\prod_{i \neq j}\PAREN{\frac{p_i}{1-p_j}}^{x_i}}{\prod_{i \neq j}x_i!}\,.
\end{align*}
\end{proof}
\end{comment}

\subsection{Proof of \Cref{lem:mst-add-err}}\label{app:mst-add-err}

    Let $E_0 \subseteq E$ denote the embedded edges where $x_{\ell(e)} = 0$.
    By the embedding procedure, we know $E_0$ is a Bernoulli subsample of $E$ with probability $\frac{1}{2}$. 
    It suffices to show that all cuts in the subgraph induced by $E_0$ are non-empty, as this will mean it is possible to construct a spanning tree entirely in $E_0$. 
    In other words, we show that each cut in the graph has at least one zero-edge sampled with high probability, which, by a standard argument, implies that there must exist a zero-cost spanning tree.
    Now let $\cE$ denote the \emph{complement} event; i.e., that there exists a cut $A,B$ such that $E_0(A,B) = 0$. 
    We will upper bound $\Pr{\cE}$ by a union bound over all cuts. 
    Let $\lambda$ denote the minimum cut size in $(V,E)$, and let $N_i$ denote the number of cuts $A,B$ such that $E(A,B) = i$. The probability that a cut with value $i$ has no edges sampled is given by $2^{-i}$. By the union bound, we have
    \begin{align*}
        \Pr{\mathcal{E}} \leq \sum_{i=\lambda}^{|E|} 2^{-i} N_i\,.
    \end{align*}

    To bound this, we use~\cite[Theorem 6.2]{karger1993global} which shows that $N_i \leq n^{\lceil 2 i / \lambda \rceil} \leq n \cdot n^{2i / \lambda}$. Plugging this in, we have by a geometric sum
    \begin{align*}
        \sum_{i=\lambda}^{|E|} 2^{-i} N_i \leq \sum_{i=\lambda}^{\infty} 2^{-i} n \cdot n^{2i/\lambda} \leq n \sum_{i=\lambda}^{\infty} \left(\frac{n^{2/\lambda}}{2}\right)^{i} &=  n\frac{\left(\frac{n^{2/\lambda}}{2}\right)^{\lambda }}{1 - \left(\frac{n^{2/\lambda}}{2}\right)} =\frac{n^{3}}{\PAREN{1 - \frac{1}{2}\left(n^{2/\lambda}\right)}\cdot 2^{\lambda}}.
    \end{align*}
    Using $\lambda \geq 5 \log(n)$, it holds that $n^3 2^{-\lambda} < \frac{1}{n^2}$, and $\frac{1}{2} \cdot n^{2/\lambda} \leq \frac{2^{2/5}}{2} < \frac{2}{3}$, so the denominator $(1-\frac{1}{2}n^{2/\lambda}) \geq \frac{1}{3}$.
    Thus, the probability of failure is at most $\frac{3}{n^2}$. \qed

\subsection{Proof of Lemma~\ref{lem:constant-fraction}}\label{app:constant-fraction}

\constantfraction*
\begin{proof}
Define $\tilde{X}_i=\mathbb{1}\left[\exists\, j\neq i: X_j=X_i\right]$ and let \(Y=\sum_{i=1}^d\tilde{X}_i\). 
By a union bound, \(\mathbb{E}[\tilde{X}_i]\le\sum_{j\neq i}\Pr{X_j=X_i}=(d-1)/n\), so that $\mathbb{E}[Y]\le d(d-1)/n = \alpha^2 n-\alpha\le\alpha^2 n$.
Since changing a single $X_i$ affects at most two of the \(\tilde{X}_i\)'s, \(Y\) is 2-Lipschitz. By McDiarmid's inequality \citep{McDiarmid_1989}, for any \(t>0\), $\Pr{Y\ge\mathbb{E}[Y]+t}\le\exp(-2t^2/(4d))$.
Setting $t=\mathbb{E}[Y]$ gives $\Pr{Y\ge2\mathbb{E}[Y]}\le\exp(-\mathbb{E}[Y]^2/(2d))=\exp(-\Omega(\alpha^3 n))$.
Thus, with probability at least \(1-\exp(-\Omega(\alpha^3 n))\), we have $Y\le2\mathbb{E}[Y]=\cO(\alpha^2 n)$.
\end{proof}

\subsection{Proof of Theorem~\ref{th:lb-matching}}\label{apx:proof-lb-mwpm}

\lbmatching*
Fix some constant $0 < \alpha \leq 1$ and let $R$ again be a parameter that controls the weights.
Now let $\Aenc_R:[n]^d\rightarrow \cG_\omega$ be constructed as follows.
Let $\vec{x}\in [n]^d$ be some vector of length  $d = \alpha n$ and $0<\alpha\leq1$.
Then the encoder creates a complete bipartite graph $G = (V_1\cup V_2, E, \vec{w})$ with bipartitions $V_1=\{\ell_1, \cdots, \ell_n\}$ and $V_2 = \{r_1, \cdots, r_n\}$.
Each of the first $d$ vertices  $\ell_i \in V_1$ represents a single coordinate of $\vec{x}$. 
For each $i \in [d]$, we set $w(\{\ell_i, r_{x_j}\}) := 0$ and
$w(\{\ell_i, r_{j}\}) := R$  for $j \neq x_i$.
%$r_j\in V_2$ we set $w(\{\ell_i, r_{x_i}\}) := 0$ to $R$ otherwise.
We treat the remaining $(1-\alpha)n$ vertices in $V_1$ as dummy nodes that simply get zero-weight edges to all vertices in $V_2$.
% \paragraph{Encoding} Let $\Aenc_R:[n]^n \rightarrow \cG_\omega$ be the encoding function that takes some dataset  $\vec{X} \in [n]^n$ and encodes them into the MST of a weighted graph $\cG_\omega$.
% The parameter $R>1$ is used to control the edge weights.
% We construct a new connected graph $G = (V_1 \cup V_2, E, \vec{w})$ with $2n$ vertices and $n^2+(n-1)$ edges in the following way: 
Furthermore, let $\Adec_G:\cP(G)\rightarrow [n]^d$ and $\Rec_G:[n]^d \rightarrow [n]^d$ be analogously defined as in \cref{ch:enc-mst}.
%Note that the decoding will always give us some permutation over $[n]$.
Find this construction also in \cref{fig:encode-as-graph}.

Using this encoding, the optimal solution is determined by the number of collisions $C$ in $\vec{x}$ and is not necessarily $0$ as before.
By \cref{lem:collision-free}, we can make $C$ arbitrarily small with probability exponentially small in $n$. 
The key observation is that there are still $\Theta(n)$ many coordinates that do not collide.
Every non-colliding coordinate can be recovered from a correct edge in some perfect matching $M$. 
Hence, we get this observation. 

\begin{observation}\label{obs:accuracy-match}
Let $G = (F, \vec{w})$ be obtained from $\Aenc_R(\x)$ and let $C = \left|\{i \in [d]| \exists j \neq i, x_i = x_j\}\right|$ be the number of colliding coordinates.
If $M \in \cP(F)$ is any perfect matching, then
\begin{align*}
  d_H(\Adec_F(M),\x) \le \frac{w(M)-w(M^*)}{R} + C,
  \end{align*}
where $M^*$ is an optimal MWPM of $G$.
\end{observation}

We now have all the ingredients ready to prove the theorem.
    The proof is similar to \cref{th:lb-mst}.
%    , but this time we only focus on a linear-sized collision free subset of the coordinates that by \cref{lem:collision-free} exists with high probability.
    First, draw $\X \sim \operatorname{Uni}([n]^{\alpha n})$ and fix a small constant $0<\alpha\leq 1$. 
    WLOG, assume that $\alpha n$ is an integer.
    Then $C \leq \alpha^2n$ with probability $1-\exp(-\Omega(\alpha^3 n))$ by \cref{lem:constant-fraction}.

    Suppose for contradiction that there exists an $(\epsilon, \delta)$-DP MWPM algorithm $\cB$ with expected additive error $\Ex{w(M) - w(M^*)} < 0.01\, \alpha n R$.
    Then by \cref{obs:accuracy-match}, $\Ex{d_H(\Adec_G(M), \X)} \leq 0.01\, \alpha n + \alpha^2 n$.
    So $\cB$ correctly recovers at least $\alpha n - 0.01\, \alpha n - \alpha^2 n = \alpha n (0.99 - \alpha)$ coordinates in expectation.
%    Then, by observation \cref{obs:accuracy-match}, we would leak more than $\alpha n (0.99 - \alpha)$ of the input coordinates in expectation. \lukasnote{Check this again.}
    
    Observe that for some fixed $R>1$, \cref{lem:induced-hamming-neighborhood} holds again, and the induced $\ell_1$-distance of two encoded graphs differs by at most $2\lceil R \rceil$ if the encoded vectors are Hamming neighbors.
    Hence, by \cref{lem:reidentification}, 
    \begin{align*}
    \Pr{\Rec_{\cA, R}(\X)_i = x_i} & \leq \frac{e^{2\lceil R\rceil\epsilon}}{n} + 2\lceil R\rceil e^{2\lceil R\rceil\epsilon}\delta\\
    & \leq \frac{1}{n} \exp\left(\ln(n^{1/c})\right) + \frac{\ln(n)}{\epsilon c} \exp\left(\ln n^{1/c}\right) \delta \\
    & = n^{(1-c)/c} + \frac{\ln(n)}{\epsilon c} n^{1/c} \delta\,.
    \end{align*}
Assuming $\delta = n^{-\Omega(1)}$, the second term is $o(1/n)$ and vanishes.
Using linearity of expectation over all $i \in [\alpha n]$,
\begin{align*}
\Ex{\sum_{i \in [\alpha n]}\mathbb{1}\left[\Rec_{\cA, R}(\vec{x})_i = x_i\right]} \leq   \PAREN{n^{(1-c)/c}\cdot \alpha n} = \alpha \cdot n^{1/c}\,,
\end{align*}
For sufficiently large $c>1$, this is sublinear in $n$, contradicting the assumed utility guarantee of $\Theta(\alpha n)$.
Finally, to obtain a bound that also depends on the number of edges, we use the same decomposition technique as in the proof of \cref{th:lb-mst} with simplification that we don't need to connect the subgraphs.

\subsection{Proof of \Cref{lem:good-cluster}}\label{app:good-cluster}

First, we will show that the connected components in the sampled edges $E'$ have size $\cO(\log n)$. 
\begin{restatable}[Largest components]{lemma}{expanderOpt}\label[lemma]{lem:small-comp}
    Suppose $F = (V, E)$ is a graph and let $H$ be the subgraph obtained by keeping each edge independently with probability $p \leq \frac{1-\alpha}{d_{\max}(F)}$ for some $\alpha \in (0,1)$. Then the maximal size of a connected component in $H$ is $\frac{2-\alpha}{\alpha^2} 4\log n$ with probability at least $1 - n^{-3}$.
\end{restatable}

We first introduce the following definitions.
\begin{definition}[Stochastic domination]
A random variable $X$ is \emph{stochastically dominated} by a random variable $Y$, if $\Pr{X \geq a} \leq \Pr{Y \geq a}$ for all $a$.
A coupling of $X$ and $Y$ is a random vector $(\tilde X, \tilde Y)$ such that the marginal distributions coincide with $X$ and $Y$ respectively.
\end{definition}

The following theorem due to \citet{Strassen1965} shows that stochastic dominance is equal to the existence of a monotone coupling.
%\url{https://people.math.wisc.edu/~roch/mdp/roch-mdp-chap4.pdf?utm_source=chatgpt.com}
\begin{lemma}[\cite{Strassen1965}]\label{lem:coupling}
A random variable $X$ stochastically dominates a random variable $Y$, if and only if there exists a (monotone) coupling $(\tilde X, \tilde Y)$ of $X$ and $Y$ such that $\Pr{\tilde X \geq \tilde Y} = 1$.
\end{lemma}

Now, we are ready to prove Lemma~\ref{lem:small-comp}.
\begin{proof}
We prove this result assuming that the graph is $d$-regular, so $d_{\max}(F) = d$. The proof for a general graph is similar.
Let $\alpha \in (0,1)$, set $p:= \frac{1-\alpha}{d}$ and let $C_v$ denote the connected component of a fixed vertex $v$ in $H$ and WLOG assume that our edge subsampling is exploring $C_v$ by BFS starting in $v$.
If $0 \leq m_t \leq d$ is the number of neighbors of the currently explored vertex at step $t \geq 1$ that have not already been explored, then observe that the size of the frontier conditioned on the past $\cF$, can be described as
\[
X_0 = 1\,, \qquad X_t = X_{t-1} -1 + Z_t \qquad \text{where }\,  Z_t\mid \cF_{t} \sim \Bin(m_t, p)\,,
\]
because we keep each neighbor that is not already included with probability $p$ independently.

Intuitively, cycles can only reduce the number of newly discovered vertices, so the BFS exploration is dominated by the process that simply ignores them.
Therefore, define the process
\begin{align}
    Y_0 = 1\,, \qquad Y_t = Y_{t-1} -1 + \tilde Z_t \qquad \text{where } \tilde Z_t \sim \Bin(d, p)\,,
\end{align}
We claim that for each $t \geq 1$, the random variable $Y_t$ conditional on the past, stochastically dominates $X_t$.
Therefore, we will explicitly give a coupling $(\tilde X, \tilde Y)$ of $X$ and $Y$ and apply \cref{lem:coupling}.
Set $\tilde X = X$ and draw $U_t\mid \cF_{t}\sim \Bin(d-m_t, p)$ independently of $Z_t$.
Now, let $\tilde Y_t = \tilde Y_{t-1} -1 + Z_t + U_t$ and consider the coupling $(\tilde X, \tilde Y)$.
Clearly, the marginal distributions match because $Z_t + U_t \sim \Bin(d, p)$ by the convolution of binomials.
We show by induction that for each $t\geq 0$, $\tilde X_t \leq \tilde Y_t$. 
The base cases $X_0, Y_0=1$ are trivial.
Now, observe that for $Z_t \sim \Bin(m_t, p)$, we have
\begin{align*}
   \tilde X_t = \tilde X_{t-1} - 1 + Z_t  &\leq \tilde Y_{t-1} - 1 + Z_t \\
   &\leq  \tilde Y_{t-1} - 1 + Z_t + U_t = \tilde Y_t\,,
\end{align*}
where the last line follows from the non-negativity of the binomial distribution.

%Fix any $t$ and condition on the past $t-1$ steps $\cF_{t-1}$ where $X_{t-1}$ and $m_{t-1}$ are fixed.
%We will show that for each $a$, we have $\Pr{X_t \leq a \mid \cF_{t-1}} \geq \Pr{Y_t \leq a \mid \cF_{t-1}}$ by induction.
%The base case is immediate.
%Since $m_t \leq d$. we have for all possible $a$, 
%\begin{align}
%     \Pr{X_t \leq a \mid \cF_{t-1}} &= \Pr{X_{t-1} + \Bin(m_t, p)- 1 \leq a\mid \cF_{t-1}}\nonumber \\
%     &\geq \Pr{X_{t-1} + \Bin(d, p) - 1 \leq a \mid \cF_{t-1}}\nonumber\\
%     &= \sum_{b = 0}^{d}\Pr{X_{t-1}\leq a - b + 1\mid \cF_{t-1}}\cdot \Pr{\Bin(d, p) = b} \label{line:ltp}\\
%     &\geq \sum_{b = 0}^{d}\Pr{Y_{t-1}\leq a - b + 1 \mid \cF_{t-1}}\cdot \Pr{\Bin(d, p) = b} \label{line:hyp} \\
%     &=\Pr{Y_t \leq a \mid \cF_{t-1}}\,,
% \end{align}
%where line \ref{line:ltp} follows by independence and the law of total probability, line \ref{line:hyp} by the induction hypothesis.\lukasnote{Maybe, one can also just skip the Law of Total probability part. I just want to make siure that we can seperate the randomnes because of independence.}

Hence, $Y_t$ gives a valid upper bound on how large the component grows, i.e.
\begin{align*}
   \Pr{|C_v| \geq t} &= \Pr{X_t \geq 1}
    \leq \Pr{Y_{t} \geq 1}\,.
\end{align*}

%Now it is left to analyse this alternative process $Y$.
Because $Y_t$ is a sum of $t$ binomials, by a standard convolution result, we get
\begin{align*}
Y_t  =  \sum_{k=1}^{t} \tilde Z_k - t \,,
\end{align*}
which is distributed as $Y_t \sim \Bin(td, p) - t$.
Now, let $\mu = t(1-\alpha)$.
By a standard Chernoff bound, ${\Pr{\Bin(td, p) \geq (1+\gamma)\mu} \leq \exp(-\frac{\gamma^2\mu}{2+\gamma})}$ for all $\gamma >0$.
Rearranging the term, yields
\begin{align}
\Pr{\Bin(td, p) \geq t } &\leq \exp\left(-\frac{(t-\mu)^2}{t + \mu}\right)
= \exp\left(-t\cdot \frac{\alpha^2}{2-\alpha}\right)\,.
\end{align}
To make this decrease with $\frac{1}{n^3}$, we choose $t = \frac{2-\alpha}{\alpha^2} \cdot 4 \ln(n)$.
Then we get by a union bound
\begin{align*}
   \Pr{\exists v: |C_v| \geq t} &\leq n \exp\left(- 4 \ln(n)\right) \\
   &= n^{-3}\,.
\end{align*}
Hence, with probability at least $1- n^{-3}$, no component has size larger than $\frac{2-\alpha}{\alpha^2}\cdot  4 \log n$.
\end{proof}

Now, we are ready to prove \Cref{lem:good-cluster}.
\goodcluster*
\begin{proof}
Let $H = (V, E_H)$ be the random subgraph consisting of the sampled edges, i.e. $E_H = \{e \in E \mid w(e) = 1\}$ and note that only edges in $E_H$ contribute to the Dasgupta cost in $G$.
This is equivalent to analyzing the connected components $C_1, \cdots, C_k\subseteq V$ of $H$ separately because the overall  Dasgupta cost decomposes to a sum of the cost within each $C_i$.

Now, observe that any $C_i$ can only contribute an additive term of at most $\lvert C_i \rvert \cdot m_i$ where $m_i$ is the number of edges in component $C_i$.
The worst case can be achieved if each of the $m_i$ edges has to cross the balanced cut, thus adding cost of at most $\lvert C_i \rvert$ each.

Therefore, we have for the optimal Dasgupta clustering tree $\cT$, 
\begin{align*}
    \hccost{w}(\cT) &\leq \sum_{i \in [k]} \lvert C_i \rvert \cdot m_i  \leq  \max_{i \in [k]} \lvert C_i\rvert\cdot  \sum_{i \in [k]} m_i = \max_{i \in [k]} \lvert C_i \rvert \cdot \lvert E_H \rvert\,.
\end{align*}

By \cref{lem:small-comp}, with probability at least $1-n^{-3}$, each $\lvert C_i \rvert \leq \frac{2-\alpha}{\alpha^2} \cdot 4 \log n$.
Furthermore, $\lvert E_H \rvert \sim \Bin(\frac{nd}{2}, \frac{1-\alpha}{d})$ and by a Chernoff bound with probability at least $1-\exp(-\Omega(n))$, we have $\lvert E_H \rvert \leq n$.
A union bound over both high probability events yields the desired result.
\end{proof}

\subsection{Proof of Lemmas~\ref{lem:cut-lb} and~\ref{lem:bc-mult-shrink}}\label{app:cut-proofs}
\begin{proof} (Of Lemma~\ref{lem:cut-lb})
    By construction, each edge $e \in E$ is independently assigned to $E \setminus E'$, $E_0'$ and $E_1'$ with probabilities $(1-p, \frac{p}{2}, \frac{p}{2})$. Thus, conditioned on $E_1'$, each remaining edge in $e \in E \setminus E_1'$ is independently assigned to $E \setminus E'$ or $E_0'$ with probability $\big(\frac{1-p}{1-p/2}, \frac{p/2}{1-p/2}\big)$. 
    Thus,
    \begin{align*}
        \Ex{E_0'(A_0,B_0) \mid E_1'} &= \frac{p/2}{1-p/2}\lvert (E\setminus E_1')(A_0,B_0)\rvert \\
        &\geq \frac{p}{2}\lvert (E\setminus E_1')(A_0,B_0)\rvert\,. \qedhere
    \end{align*}
\end{proof}

\begin{proof}(Of Lemma~\ref{lem:bc-mult-shrink})
    $E_1'(A_0, B_0)$ is a binomial distribution on $E(A_0,B_0)$ with probability $\frac{p}{2}$. Thus, it suffices to compute a tail bound $\Pr{X \geq \frac{m}{3}}$, where $m = E(A_0,B_0)$ and $X \sim \Bin(m,\frac{p}{2})$. By Chernoff's bound, we know $\Pr{X \geq \frac{mp}{2} + t} \leq \exp(-\frac{t^2}{mp + 2t/3})$ for all $t 
    \geq 0$. 
    By assumption we know $p = \probp \leq \frac{1}{8}$. 
    Taking $t = \frac{m}{4}$ the above bound is $\Pr{X \geq \frac{m}{16} + \frac{m}{4}} \leq \exp\big(-\frac{m^2/16}{m/4 + m/6}\big) < \exp(-\frac{m}{8})$. Because $m \geq \phi(F) \geq 8n$, we have $\exp(-\frac{m}{8}) \leq \exp(-n)$, the failure probability over all $2^{n}$ balanced cuts is at most $2^n e^{-n} \leq 2^{-0.2n}$.
\end{proof}

\subsection{Proof of Theorem~\ref{th:ub-mst}}\label{app:ub-mst}
\ubmst*
    The following proof is due to~\cite{aamand-personal-communication}. 
    Consider the algorithm that adds noise $Z_e \sim \lap(1/\epsilon)$ to each edge independently.
    Denote these weights with $\vec{w}'$.
    This mechanism is $\epsilon$-DP under the $\ell_1$ neighboring relationship by a standard result \citep{Dwork2006}.
    Now let $T\subseteq \cS(G)$ be an optimal solution over the noisy weights $\vec{w}'$ and $T^*$ be an optimal over the original weights $\vec{w}$.
    Let $Z_T = \sum_{t \in T} Z_t$ be the total noise of any $T$ and let $M = \max_{T \in \cS(G)} |Z_T|$.
    Then,
    \begin{align*}
        w(T') \leq w'(T') + M  \leq w'(T^*) + M \leq w(T^*) + 2 \cdot M\,,
    \end{align*}
    by the same inequalities as used in \cite{Sealfon_2016}.
    Note that the sub-exponential norm of each of the i.i.d.~Laplace random variables $Z_e$ is $\|Z_e\|_{\psi_1} \leq \frac{2}{\epsilon}$.
    By a standard Bernstein inequality \citep{Vershynin2018}, we get for each $T$ for any $t >0$ and an absolute constant $c >0$,
    \begin{align*}
        \Pr{|Z_T| \geq t} \leq  2\exp\left(-c\,\min\left(\tfrac{\epsilon^2 t^2}{n},\; \epsilon t\right)\right),
    \end{align*}
    Now setting $t = c \cdot (k/\epsilon) \ln(m/k)$ for a sufficiently large constant $c$, the tail above becomes dominated by the linear term $\exp\big(-\Omega(n \ln(m/k))\big)$.
    Clearly,  $|\cS(G)|  \leq {n \choose k} \leq \PAREN{\frac{em}{k}}^k.$
%    \frac{1}{n}\PAREN{\frac{2m}{n-1}}^{n-1} \leq \PAREN{\frac{2m}{n-1}}^{n-1}$.
    A union bound over $\cS(G)$ yields,
    \begin{align*}
        \Pr{\max_{T \in \cS(G)} |Z_T| \geq C \dfrac{k}{\epsilon}\ln(m/k)} \leq n^{-\Omega(1)}\,,
    \end{align*}
    Combined with the chain of inequalities above, this gives $w(T') - w(T^*) \leq \cO((k/\epsilon)\ln(m/k))$ with probability at least $1 - n^{-\Omega(1)}$.

%% file: fig/clustering-example.tikz
\begin{tikzpicture}
	\begin{pgfonlayer}{nodelayer}
		\node [style=BLACKLARGE] (0) at (-3.75, 2) {};
		\node [style=BLACKLARGE] (1) at (-1.25, 1.5) {};
		\node [style=BLACKLARGE] (2) at (-3, 0) {};
		\node [style=BLACKLARGE] (3) at (-5.25, 0.25) {};
		\node [style=BLACKLARGE] (4) at (2, 0) {};
		\node [style=BLACKLARGE] (5) at (4, 0) {};
		\node [style=BLACKLARGE] (8) at (3, 1.5) {};
		\node [style=BLACKLARGE] (12) at (4.5, 2.75) {};
		\node [style=LIGHTRED] (13) at (5, 0) {};
		\node [style=LIGHTRED] (14) at (7, 0) {};
		\node [style=LIGHTRED] (15) at (6, 1.5) {};
		\node [style=none] (17) at (1.75, -1) {$v_1$};
		\node [style=none] (20) at (3.75, -1) {$v_2$};
		\node [style=none] (21) at (5, -1) {$v_3$};
		\node [style=none] (22) at (7, -1) {$v_4$};
		\node [style=none] (23) at (-5.75, 1) {$v_1$};
		\node [style=none] (24) at (-3.75, 2.75) {$v_2$};
		\node [style=none] (25) at (-2.25, -0.25) {$v_3$};
		\node [style=none] (26) at (-1.25, 2.25) {$v_4$};
		\node [style=none] (29) at (-2.5, 1) {$e$};
		\node [style=none] (30) at (0, 2.75) {};
		\node [style=none] (31) at (0, 0) {};
		\node [style=none] (32) at (7.5, 2.4) {$\operatorname{lca}(v_3, v_4)$};
		\node [style=none] (33) at (2.5, 3.25) {$\mathcal{T}$};
		\node [style=none] (34) at (-5.5, 3.25) {$G$};
	\end{pgfonlayer}
	\begin{pgfonlayer}{edgelayer}
		\draw [style=EDGE] (3) to (0);
		\draw [style=EDGE] (0) to (2);
		\draw [style=FILLLIGHTREDB] (2) to (1);
		\draw [style=EDGE] (1) to (0);
		\draw [style=EDGE] (8) to (4);
		\draw [style=EDGE] (8) to (5);
		\draw [style=EDGE] (8) to (12);
		\draw [style=FILLLIGHTREDB] (15) to (13);
		\draw [style=FILLLIGHTREDB] (15) to (14);
		\draw [style=EDGE] (15) to (12);
		\draw [style=EDGEDASHED] (30.center) to (31.center);
	\end{pgfonlayer}
\end{tikzpicture}